\documentclass[article]{autart}

\usepackage{wrapfig}
\usepackage{bbold}
\usepackage{enumerate}
\usepackage{cite}
\usepackage{color,xcolor}
\usepackage{amsmath}
\usepackage{amssymb}
\usepackage{graphicx}
\usepackage{algorithm}
\usepackage{algpseudocode}
\usepackage{amsmath}
\usepackage{bbm}
\usepackage{graphics}
\usepackage{subfigure}
\usepackage{epsfig}
\usepackage{multicol} 
\usepackage{amsmath}
\newtheorem{theorem}{Theorem}{}
{}
\newtheorem{lemma}{Lemma}{}
\newtheorem{assumption}{Assumption}
\newtheorem{corollary}{Corollary}{}
\newtheorem{remark}{Remark}{}
\newtheorem{problem}{Problem~}
\newtheorem{definition}{Definition~}
\newenvironment{proof}{\hspace{0ex}\textsc{Proof}.\hspace{1ex}}{\hfill$\Box$\newline}

\begin{document}

\begin{frontmatter}

\title{
	 \textcolor{black}{On the Effects and Optimal Design of Redundant Sensors in Collaborative State Estimation} \thanksref{mytitle}}
\thanks[mytitle]{The work by Y. Ren and Z. Duan  is supported by the National Natural Science Foundation of China under Grant T2121002 62173006, and the work by Y. Ren is also supported by the project Peking University-BHP Carbon and Climate Wei-Ming PhD Scholar.\emph{(Corresponding author: Zhisheng Duan.)}}

\author[mymainaddress]{Yunxiao~Ren}\ead{renyx@pku.edu.cn},    
\author[mymainaddress]{Zhisheng~Duan}\ead{duanzs@pku.edu.cn},               
\author[mysecondaryaddress]{Peihu Duan}\ead{duanpeihu.work@gmail.com},  
\author[mythirdaddress]{Ling Shi}\ead{eesling@ust.hk}  

\address[mymainaddress]{the State Key Laboratory for Turbulence and Complex Systems, Department of Mechanics and Engineering Science, College of Engineering, Peking University, Beijing, 100871, China}
\address[mysecondaryaddress]{School of Electrical Engineering and Computer Science, KTH Royal Institute of Technology, Stockholm, Sweden}
\address[mythirdaddress]{the Department of Electronic and Computer Engineering, the Hong Kong University of Science and Technology, Clear Water Bay, Kowloon, Hong Kong, China}

\begin{keyword}
 Redundant sensors, Optimal sensor design, Collaborative state estimation, Kalman filter.
\end{keyword}

\begin{abstract}
The existence of redundant sensors in \textcolor{black}{collaborative state estimation} is a common occurrence, yet their true significance remains elusive.  This paper comprehensively  investigates the effects and optimal design of redundant sensors in sensor networks that use Kalman filtering to estimate the state of a random process collaboratively. The paper presents two main results: a theoretical analysis of the effects of redundant sensors and an engineering-oriented optimal design of redundant sensors. In the theoretical analysis, the paper leverages Riccati equations and Symplectic matrix theory to unveil the explicit role of redundant sensors in cooperative state estimation. The results unequivocally demonstrate that the addition of redundant sensors enhances the estimation performance of the sensor network, aligning with the principle of ``more is better". \textcolor{black}{Moreover, the paper establishes a precise sufficient and necessary condition to assess whether the inclusion of redundant sensors improves the overall estimation performance.} Moving towards engineering-oriented design optimization, the paper proposes a novel algorithm to tackle the optimal design problem of redundant sensors, and the convergence of the proposed algorithm is guaranteed. Numerical simulations are provided to demonstrate the results.
\end{abstract}

\end{frontmatter}

\section{Introduction}
A sensor network is a collection of sensors that are connected through either a serial bus or wireless communication. This type of network facilitates efficient collaboration for sensing tasks in complex or large-scale industrial systems. Over the past few decades, the applications of sensor networks have gained significant attention. Examples include lidars on intelligent vehicles for self-driving \cite{akiyama2022edge}, sensors for environmental monitoring \cite{silva2015experimental}, and pressure and temperature monitoring \cite{silveira2016temperature} for process control.

 In real-world scenarios, a large number of sensors are often necessary to achieve comprehensive and reliable observation of complex systems or to provide full coverage of large-scale systems. Consequently, the deployment of redundant sensors is almost inevitable. Recently, researchers have expressed interest in the phenomenon of redundancy in sensor networks, but an in-depth analysis is still lacking. 
 \subsection{Related Works}
Existing research on redundant sensors in the literature primarily focuses on two aspects. Firstly, there is a focus on scheduling the sensor network to conserve energy by detecting and eliminating redundant sensors \cite{9512811,10.1007/978-981-19-2719-5_57,9123882}. 
 However, while redundant sensors can be eliminated or deactivated in certain environmental monitoring tasks, in industrial systems where safety is paramount, deactivating them may lead to fatal consequences if key sensors fail \cite{ZHANG2021100055}. Therefore, redundant sensors often have to remain active within the sensor network. Nonetheless, the aforementioned studies primarily focus on factors such as observation ranges or data similarities among sensors, disregarding the dynamics of sensors that are prevalent in real-world applications, such as sensors collaboratively performing Kalman filtering for state estimation.

The second research direction involves utilizing redundant sensors in a sensor network to compensate for sensor failures \cite{9639479,6127570,LIU2018458,9695499} or against cyber-attacks \cite{6882789,7805147}, which can improve the robustness of the entire sensor network. 
 Although these works demonstrate that redundant sensors can be used to improve the robustness of the sensor network to sensor failures and cyber-attacks, they have not quantitatively analyzed the impact of redundant sensors on the performance of the sensor network, particularly when the sensor network is collaboratively performing Kalman filtering for state estimation. Many previous studies have examined the sensor network-based Kalman filter \cite{BATTISTELLI201875,4434303,YANG2023110690,5438273,DUAN2020109007}. However, they have not explored the effects of redundant sensors either, which are ubiquitous and may influence the estimation performance.

\textcolor{black}{There are also some previous works on sensor design. In reference \cite{sensor_design1}, the investigation focuses on the sensor and actuator placement, with the objective of selecting a subset from a finite set of possible placements to optimize various real-valued controllability and observability metrics within the network. While \cite{sensor_design1} considers several metrics of the Gramian, it omits the error covariance matrix—a direct metric for evaluating state estimation performance. Reference \cite{sensor_design2} explores the co-design of control and sensing in Linear Quadratic Gaussian (LQG). The problem is formulated as the selection of sensors from a finite set, with each sensor associated with a distinct cost. In \cite{sensor_design3}, the research addresses the joint problem of selecting a set of control nodes and designing a control input for complex networks. The control nodes are selected from the entirety of the network's nodes, considering control energy. In \cite{sensor_design4,sensor_design5} 
, the optimal sensor design problem is investigated. However, the methods proposed in \cite{sensor_design4,sensor_design5}  are gradient-based, requiring a complete redesign of all sensors for any new task. In practical scenarios, existing industry sensor networks are typically well preset, and users often only need to add a few new redundant sensors to enhance network robustness or implement new features. In such situations, the cost of redesigning all sensors is prohibitive, necessitating the development of an algorithm to offer an optimal design for the plug-and-play redundant sensors to be added into the network.}
	
\textcolor{black}{ Despite the progress of past work in sensor design\cite{sensor_design1,sensor_design2,sensor_design3,sensor_design4,sensor_design5}, there are still some issues that deserve our attention in this paper. Firstly, \cite{sensor_design1,sensor_design2,sensor_design3,sensor_design4,sensor_design5} only considers scalar metrics, recognizing that investigating scalar metrics may result in information loss. For instance, in the case of two positive-definite matrices $P_{1}$ and $P_{2}$, $\operatorname{tr}(P_{1})>\operatorname{tr}(P_{2})$ does not imply $P_{1}>P_{2}$. While our exploration encompasses the matrix itself, which can reflect the overall estimation performance. Additionally, while prior works \cite{sensor_design1,sensor_design2,sensor_design3} focus on selecting a subset of sensors or actuators from a finite set, posing combinatorial optimization problems, this paper concentrates on designing sensors under specific constraints by solving the semi-definite programming problem. This approach implies that the sensors are not merely selected from a set; instead, their configurations are obtained through optimization. Lastly, the previous sensor placement strategy  \cite{sensor_design4,sensor_design5} assumes a predefined set of sensors and doesn't consider scenarios where new sensors may be added or replaced dynamically. In environments where plug-and-play sensors are prevalent, the previous method becomes impractical as it requires a complete redesign when new sensors are introduced.}

Since redundant sensors are prevalent and existing studies on their effects and optimal design are insufficient, it is indeed intriguing to investigate the impact of redundant sensors in collaborative state estimation. It is particularly valuable to explore the performance gap between sensor networks with and without redundant sensors. Furthermore, it is worth delving into the optimization of redundant sensor designs to enhance state estimation accuracy. By conducting such research, we can gain insights into the potential benefits of redundant sensors and develop strategies to optimize their configurations for improved state estimation.
\subsection{Contributions}
Motivated by the aforementioned observation, this paper aims to investigate the effects and optimal design problem of redundant sensors in cooperative Kalman filter. The contributions of this paper can be summarized as follows:

\begin{enumerate}

\item[1)] \textcolor{black}{\textbf{Effect Analysis:} This paper explores the impact of redundant sensors on sensor networks, providing a quantitative analysis of their effects on the error covariance matrix in state estimation. Unlike previous works\cite{sensor_design1,sensor_design2,sensor_design3} that focus on scalar quantities like the smallest eigenvalue and traces, our analysis concentrates on the error covariance matrix itself, preventing information loss.  \textbf{Theorem 1} establishes that sensor networks with redundant sensors can enhance estimation performance, demonstrating the principle of ``more is better." Specifically, it proves that the steady-state estimation error covariance of a network with redundant sensors is less than or equal to that of a network without redundant sensors. The paper also introduces a condition determining when adding redundant sensors strictly reduces the mean square error (MSE). This analysis establishes a novel mathematical theory for examining the effect of redundant sensors on the system.}

\item[2)] \textcolor{black}{\textbf{Condition on Improving the Overall Estimation Performance:} This paper presents \textbf{Theorem 2}, delivering a comprehensive sufficient and necessary condition to determine whether the inclusion of redundant sensors strictly enhances overall estimation accuracy. In contrast to prior works \cite{sensor_design1,sensor_design2,sensor_design3,sensor_design4} that emphasize scalar metrics and provide non-strict monotonicity results, our analysis focuses on the error covariance matrix itself, establishing a strict monotonicity condition. This condition plays a pivotal role in assessing the overall improvement in  estimation accuracy resulting from the incorporation of redundant sensors. The theorem provides invaluable insights into the influence of redundant sensors.}

\item[3)]\textcolor{black}{\textbf{Optimization Algorithm}: The paper presents an algorithm addressing the optimal design problem of redundant sensors, outlined in \textbf{Theorem 3} and \textbf{Algorithm 1}. Unlike previous works \cite{sensor_design1,sensor_design2,sensor_design3} selecting sensors from a finite set, our method directly designs sensors through optimization. The proposed algorithm rigorously demonstrates convergence and optimality in \textbf{Theorem 4} and \textbf{Theorem 5}. Notably, this algorithm distinguishes itself from existing sensor design methods \cite{sensor_design5,sensor_design4} by being applicable for plug-and-play sensors, eliminating the need for redesigning the entire sensor network.}
\end{enumerate}

 \textbf{Notation:} Throughout this paper, we adopt the following notation conventions. $\mathbb{R}$ and $\mathbb{C}$ denote the sets of real numbers and complex numbers, respectively. $\mathbb{N}$ denotes the set of natural numbers. For a complex number $\lambda$, $\lambda^{\star}$ denotes its conjugate. $\mathbb{R}^{n}$ and $\mathbb{R}^{n\times n}$ denote the sets of column vectors and square matrices of dimension $n$, respectively. \textcolor{black}{$I_{n}\in \mathbb{R}^{n\times n}$ denotes the identity matrix of dimension $n$, and $e_{i,n}$ is the $i$th column in $I_{n}$. $I$ denotes the identity matrix of appropriate dimension} The scalar zero and zero matrix are denoted by 0. For any matrix $M$, $M^{T}$, $M^{H}$, and $M^{-1}$ denote its transpose, conjugate transpose, and inverse, respectively. The Moore-Penrose inverse of $M$ is denoted as $M^{\dagger}$. For a symmetric matrix $X$, $\sqrt{X}$ denotes its Cholesky decomposition, and $X>0$ ($X\geq0$) indicates that $X$ is positive definite (positive semidefinite). The notation $X \geqq 0$ implies that $X$ is positive semidefinite, and at least one eigenvalue of $X$ is zero. The operators $\|\cdot\|_{1}$, $\|\cdot\|_{2}$, and $\|\cdot\|_{\infty}$ denote the 1-norm, 2-norm, and infinity-norm, respectively. For functions $g$, $h$ with appropriate domains, $g \circ h(x)$ stands for the function composition $g(h(x))$. For matrices $X_{1}, X_{2},\ldots, X_{n}$, $\text{diag}\{X_{1}, X_{2},\ldots, X_{n}\}$ represents the corresponding block diagonal matrix. tr($\cdot$) denotes the trace operator. \textcolor{black}{ $\delta_{k j}$ denotes the Dirac delta function equals 1 if $k=j$ and 0 ortherwise. $\operatorname{Ker}(\ldots)$ denotes the kernel operator of a matrix. $\operatorname{dom}(Ric)$ is the set of a special kind of  symplectic matrix defined in Definition 2}
\section{Problem Formulation and Preliminaries}

\subsection{Modeling Formulation}

Consider the following discrete-time linear time-invariant process:
\begin{equation}\label{eq:process}
	x_{k+1} = Ax_{k}+w_{k}.
\end{equation}
In this context, the time stamp $k \in \mathbb{N}$ represents the current time. The system matrix is denoted by $A$, and the state vector at time $k$ is represented as $x_{k} = \left[x_{1,k},x_{2,k},\ldots,x_{n,k}\right]^{T} \in \mathbb{R}^{n}$, and the state comprises $n$ channels. Each channel is denoted by $x_{i,k} \in \mathbb{R}$, representing the $i$th element of $x_{k}$. The noise term $w_{k} \in \mathbb{R}^{n}$ is a zero-mean, independent, and identically distributed white Gaussian noise with a covariance matrix $Q\geq 0$. The covariance matrix of the noise is $\mathbb{E}\left[w_{k} w_{j}^{T}\right]=\delta_{k j} Q$. Moreover, the pair $(A, \sqrt{Q})$ is controllable.

The sensor network model is described as follows:
\begin{equation}\label{eq:original_sensornet}
	\bar{y}_{k}=\bar{C} x_{k}+\bar{v}_{k},
\end{equation}

where $\bar{y}_{k} = \left[\begin{array}{ccc}y^{T}_{1,k},&\ldots,& y^{T}_{N,k}\end{array}\right]^{T}\in \mathbb{R}^{\bar{m}}$, and $y_{i,k}\in \mathbb{R}^{m_{i}}$ is the measurement of sensor $i$, $N$ is the number of sensors, $\bar{C}=\left[\begin{array}{ccc}C^{T}_{1} &\ldots &C^{T}_{N}\end{array}\right]^{T}\in \mathbb{R}^{n\times\bar{m}}$, $C_{i}\in \mathbb{R}^{n\times m_{i}}$ is the output matrix associated with sensor $i$, $\bar{v}_{k} = \left[\begin{array}{ccc}
v^{T}_{1, k} &
\ldots &
v^{T}_{N, k}
\end{array}\right]^{T}\in \mathbb{R}^{\bar{m}}$, and $v_{i,k} \in \mathbb{R}^{m_{i}}$ is zero-mean independent and identically distributed white Gaussian noise with covariance matrix $R_{i}>0$. Thus, $\bar{v}_{k}\sim N\left(0,\bar{R}\right)$, where $\bar{R}=\text{diag}\{R_{1},\ldots,R_{N} \}>0$. The following assumptions are made:

\begin{assumption}\label{assum:inver}
	$A$ is invertible.
\end{assumption}
\begin{assumption}\label{assum:obs}
	The pair $\left(A,\bar{C}\right)$ is observable.
\end{assumption}
\textcolor{black}{
For most physical plants, Assumption \ref{assum:inver} is justified because system (1) can be regarded as the discretized form of a differential equation $\text{d} x = A_1 x + A_2 \text{d} w $ such that $A \triangleq  e^{A_1 h} $ is always invertible for any $A_1$. Assumption \ref{assum:obs} means that the process is collectively observable by the sensor network, that is the whole sensor network is observable, and sensors can estimate the process state cooperatively. }

By using the consensus on measurement-based distributed Kalman filtering algorithm \cite{Distributed4SteadyKalmanFilter,BATTISTELLI2016169,QianAuto1,9929319} or equipping a fusion center, the sensor network can perform the Kalman filtering cooperatively.

\begin{equation}\label{eq:kalman_filter}
\left\{\begin{array}{ll}{\hat{x}_{k | k-1}} & {=A \hat{x}_{k-1 | k-1}} \\ {P_{k | k-1}} & {=h\left(P_{k-1|k-1} \right)} \\ {K_{k} } & {=P_{k | k-1}  \bar{C}^{T}\left[\bar{C} P_{k | k-1}  \bar{C}^{T}+\bar{R}\right]^{-1}} \\ {\hat{x}_{k|k} } & {=\hat{x}_{k | k-1} +K _{k}\left(\bar{y}_{k}-\bar{C}\hat{x}_{k | k-1} \right)} \\ {P_{k|k} } & {={g}_{\bar{C},\bar{R}}\left(P_{k | k-1} \right),}\end{array}\right. 
\end{equation}
where $\hat{x}_{k | k-1} $ and $\hat{x}_{k|k} $ are the priori and posteriori minimum mean squared error estimates of the process state $x_{k}$, and $P_{k|k-1} $ and $P _{k|k}$ are the corresponding priori and posteriori error covariance matrices, which are defined as follows: 
\begin{equation}
P_{k|k-1}  = \mathbb{E}\left[(\hat{x} _{k|k-1}-x_{k})^{T}(\hat{x} _{k|k-1}-x_{k})\right],
\end{equation}
\begin{equation}
P_{k|k}  = \mathbb{E}\left[(\hat{x} _{k|k}-x_{k})^{T}(\hat{x} _{k|k}-x_{k})\right],
\end{equation}

With the Assumptions 1 and 2,  the priori and posteriori estimation error covariance matrices of the original network converge to unique steady-state values. We denote the steady-state value of the priori and posteriori estimation error covariance matrices of the original network as $\bar{P}$ and $\bar{P}_{p}$, respectively.
\begin{equation}
\bar{P}=\lim _{k \rightarrow \infty} P_{k| k-1} ,
\end{equation}
\begin{equation}
\bar{P}_{p}=\lim _{k \rightarrow \infty} P_{k} ,
\end{equation}
 $\bar{P}$ is determined by the unique positive definite solution of the following discrete-time algebraic Riccati equation (DARE):
\begin{equation}\label{eq:DARE1}
	\bar{P}=A \bar{P} A^{T}+Q-A \bar{P} \bar{C}^{T}\left(\bar{C} \bar{P} \bar{C}^{T}+\bar{R}\right)^{-1} \bar{C} \bar{P} A^{T},
\end{equation} 
which is also the solution of the equation $h\circ g_{\bar{C},\bar{R}}(\bar{P}) = \bar{P}$. Similarly, $\bar{P}_{p}$ is determined by the unique positive definite solution of the equation $g_{\bar{C},\bar{R}} \circ h(\bar{P}_{p}) = \bar{P}_{p}$, where $h$ and $g_{\bar{C},\bar{R}}$ are the Lyapunov and Riccati operators: $\mathrm{S}_{+}^{n} \rightarrow \mathrm{S}_{+}^{n}$ defined as follows:
\begin{equation}
	h(X) \triangleq A X A^{T}+Q,
\end{equation}
\begin{equation}
	g_{C,R}(X) \triangleq X-X C^{T}[ C X C^{T}+R]^{-1} C X.
\end{equation}
\textcolor{black}{\begin{remark}\textbf{The convergence and stability of the investigate estimate method:}
		It shoud be noted that this paper primarily investigates the effects and optimal design of redundant sensors in collaborative state estimation, i.e., we focus on the performance gap between sensor networks with and without redundant sensors being added, where the sensors in the network collaboratively execute a global Kalman filter which remains the same as the centralized approach. Therefore, we do not develop any new estimation strategy in this paper. Instead, the considered collaborative state estimation can be fulfilled by many previous literature\cite{Distributed4SteadyKalmanFilter,QianAuto1,RYU2023110843}, and the convergence and stability of these methods are strictly proved. Therefore, in this context, a sensor network denotes a group of sensors collaborating in their operation, the communication and local filter design is not our emphasis and is omitted.
\end{remark}} 
 The error covariance matrices converge to steady-state values exponentially, and the steady-state covariance matrices reflect the estimation performance in the steady state. Therefore, in the following, steady-state covariance matrices are analyzed.  \par
\subsection{Problem Statement} 
Redundant sensors are commonly present in various industrial processes, where a subset of sensors already provides collective observability and is capable of estimating the process state. However, the entire sensor network may contain additional redundant sensors beyond this subset. The impact of these redundant sensors on the collaborative state estimation performance of the sensor network is uncertain. While redundant sensors can be removed if they are deemed useless, they may actually have an effect on the estimation performance of the sensor network. This raises the question of how different configurations of redundant sensors can result in varying estimation performances. Therefore, studying the effects of redundant sensors is an intriguing and valuable pursuit. Furthermore, if redundant sensors are found to be beneficial, it becomes important to investigate how to optimally design them to achieve improved performance.

\textcolor{black}{By regarding the sensor network \eqref{eq:original_sensornet} as the original sensor network, we can add the following redundant sensors into the network:}

\begin{equation}
	\tilde{y}_{k} = \tilde{C} x_{k}+\tilde{v}_{k},
\end{equation}
where 
$\tilde{y}_{k}= \left[\begin{array}{ccc}
	y^{T}_{N+1,k}&\ldots&y^{T}_{N+m,k}
\end{array}\right]^{T}\in \mathbb{R}^{\tilde{m}}$ is stacked output of redundant sensors, $y_{N+i,k} \in\mathbb{R}^{m_{N+i}}$ is the output of $i$th redundant sensor,
$\tilde{C} =\left[\begin{array}{ccc}
		C^{T}_{N+1} &
		\ldots &
		C^{T}_{N+m}
	\end{array}\right]^{T}\in \mathbb{R}^{n\times \tilde{m}}$
is the collection of the $m$ redundant sensors' output matrices, $C_{N+i}$ is the output matrix of $i$th redundant sensor. $\tilde{v}_{k}= \left[\begin{array}{ccc}
v^{T}_{N+1, k} &
\ldots&
v^{T}_{N+m, k}
\end{array}\right]^{T}$ is the  collection of the corresponding measurement noises with covariance $\tilde{R}=diag\{R_{N+1},\ldots,R_{N+m} \}$. Denote $C = \left[\begin{array}{cc}
\bar{C}^{T}&\tilde{C}^{T}
\end{array}\right]^{T}$, $v_{k} = \left[\begin{array}{cc}
\bar{v}_{k}^{T}&\tilde{v}_{k}^{T}
\end{array}\right]^{T}$, $R = diag\{\bar{R},\tilde{R}\}$. 

\textcolor{black}{For simplicity, in the rest part of the paper, we use the output matrix $\bar{C}$, $\tilde{C}$ and $C$ to denote the original sensor network, the redundant sensors and the sensor network with the redundant sensors added, respectively.}

Obviously, with Assumption \ref{assum:obs}, the pair $(A,C)$ is also observable. The priori and posteriori covariances corresponding to redundant network $C$ will converge to the solutions of the following equations, respectively.
\begin{equation}\label{eq:DARE2}
P=A P A^{T}+Q-A P C^{T}\left(C P C^{T}+R\right)^{-1} C P A^{T},
\end{equation}
\begin{equation}
g_{C,R}\circ h(P_{p}) = P_{p}.
\end{equation}

\begin{figure}
	\centering
	\includegraphics[width=1\linewidth]{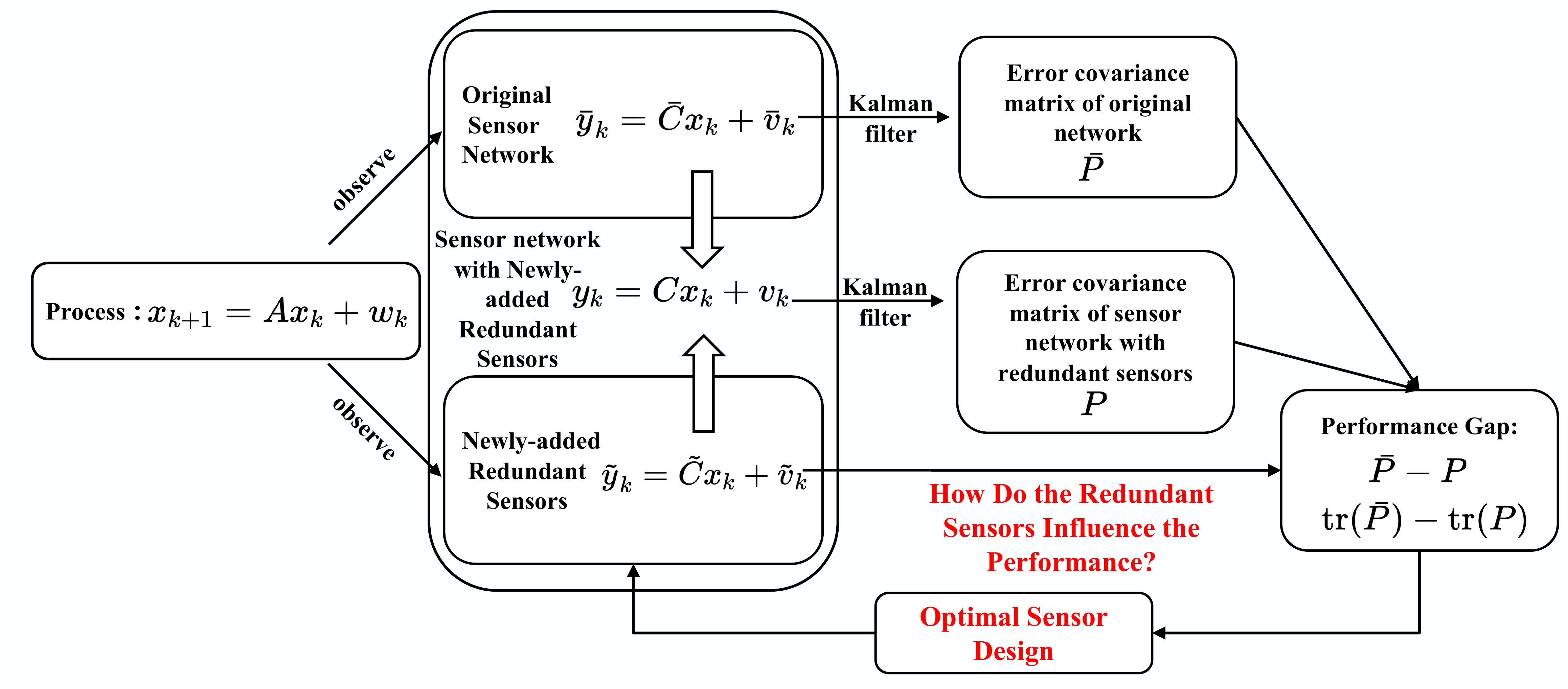}
	\caption{\textcolor{black}{The model and research problem concerned.}}
	\label{fig:structure}
\end{figure}

In this paper, two key questions are raised for investigation:
\begin{enumerate}
	\item[(1)] What are the effects of redundant sensors on the error covariance of the estimation?
	\item[(2)]  How can engineers design redundant sensors to achieve a better estimation performance?
\end{enumerate}
Fig. \ref{fig:structure} illustrates the modeling formulation and concerned problem in this paper. The two questions can be addressed by analyzing the differences in estimation performance between the original sensor network and the sensor network that includes redundant sensors.
\textcolor{black}{
To formulate the above two questions into mathematical form, the following preparations are made. In the context of this paper, the MSE is defined as the expected value of the squared difference between the estimated state $\hat{x}_k$ and the true state $x_k$:}
\begin{equation}
	MSE = \mathbb{E}\left[(\hat{x}_{k}-x_{k})^{T}(\hat{x}_{k}-x_{k})\right].
\end{equation}
 This MSE can also be expressed as the trace of the corresponding error covariance matrix.

In the Kalman filter, there are two types of covariance matrices involved: the priori covariance and the posteriori covariance. Without loss of generality, this paper specifically focuses on analyzing the prior covariance matrix. Particularly, the prior covariance matrix has a direct relationship with the solution of a DARE, allowing for the utilization of various techniques to facilitate the analysis.The following lemma illustrates the relationship between the priori and posteriori covariance matrices.
\textcolor{black}{
\begin{lemma}\label{lemma:congruence}
	If Assumptions \ref{assum:inver} holds. For process \eqref{eq:process}, and two different sensor networks SN1, SN2 with output matrices and noise covariance matrices $(C_{1},R_{1})$ and  $(C_{2},R_{2})$ respectively, denote the corresponding priori and poesteriori covariance matrices as $P_{1}$, $P_{p,1}$ and $P_{2}$ and $P_{p,2}$. Then one has  $P_{1}\geq P_{2}$ if and only if $P_{p,1}\geq P_{p,2}$, and $P_{1}> P_{2}$ if and only if $P_{p,1}> P_{p,2}$, respectively.
\end{lemma}}
\begin{proof}
	From \eqref{eq:kalman_filter}, it can be obtained that 
	\begin{equation}
	P_{i} = AP_{p,i}A^{T}+Q ,~ i=1,2.
	\end{equation}
	Hence, one has 
	\begin{equation}\label{eq:congruence}
	P_{1}-P_{2} = A(P_{p,1}-P_{p,2})A^{T}.
	\end{equation}
	Since $A$ is invertible, \eqref{eq:congruence} means that $P_{1}-P_{2}$ and $P_{p,1}-P_{p,2}$ are congruent. From Sylvester's law of inertia, it follows that  $P_{1}-P_{2}$ and $P_{p,1}-P_{p,2}$ have the same numbers of positive, negative, and zero eigenvalues. This completes the proof.
\end{proof}
\begin{corollary}\label{corollary:con}
	If $P_{1}\geq P_{2}$, $tr(P_{1})\geq (>) tr(P_{2})$ if and only if $tr(P_{p,1})\geq (>) tr(P_{p,2})$.
\end{corollary}
By utilizing Lemma \ref{lemma:congruence} and Corollary \ref{corollary:con}, it can be deduced that analyzing the priori covariance matrix is equivalent to analyzing the posteriori covariance matrix. 

Therefore, in the subsequent sections of this paper, the comparison of priori error covariance matrices between the original sensor network and the sensor network with redundant sensors provides insights that hold for the posteriori matrices as well. With above preparation, two main problems in this paper are set up:
\textcolor{black}{
\begin{problem}\label{prob:1} Consider the estimation of the state $x_{k}$ in the process \eqref{eq:process} using the original sensor network $\bar{C}$, resulting in the error covariance matrix $\bar{P}$. Similarly, when utilizing the redundant sensor network $C = \left[\bar{C}^{T}~\tilde{C}^{T}\right]^{T}$ for the same estimation, the obtained error covariance matrix is $P$. Determine the specific relationship between $\bar{P}$ and $P$, such as $\bar{P}\geq P$, $\bar{P}\leq P$, $\bar{P} = P$, or $\bar{P} - P$ is indefinite. Furthermore, establish the quantitative relation between $\text{tr}(\bar{P})$ and $\text{tr}(P)$. \end{problem} }
\textcolor{black}{
Based on the relationship analysis in Problem 1, how can we optimize the estimation performance $P$ with respect to the configuration of redundant sensors, i.e., $\tilde{C}$? Specifically, one more  optimization problem for design optimization is formulated as follows: minimize the trace of the covariance matrix subject to the constraints that the covariance matrix satisfies the DARE \eqref{eq:DARE1} and the redundant sensors are restricted to a convex set due to the practical sensor design consideration, such as the norm bounded set.}
\textcolor{black}{
\begin{problem}\label{prob:2}(Optimal Design) Determine the optimal configuration $\tilde{C}^{*}$ for redundant sensors by minimizing the associated error covariance matrix:
	\begin{equation}
\tilde{C}^{*}	=\arg\min_{\tilde{C}} tr(P), ~ s.t.
	\end{equation}
	\begin{equation}\label{eq:p2co}
	\begin{array}{c}
	P=A P A^{T}+Q-A P \\ \times\left[\begin{array}{c}
	\bar{C}\\\tilde{C}
	\end{array}\right]^{T}\left(\left[\begin{array}{c}
	\bar{C}\\\tilde{C}
	\end{array}\right] P \left[\begin{array}{c}
	\bar{C}\\\tilde{C}
	\end{array}\right]^{T}+R\right)^{-1} \left[\begin{array}{c}
	\bar{C}\\\tilde{C}
	\end{array}\right] P A^{T},\\
	\tilde{C} \in U,~
	\end{array}
	\end{equation}
	where $U$ is the available set of redundant sensors. 
\end{problem} }
\textcolor{black}{
\begin{remark}
	In practical scenarios, designing redundant sensors is subject to constraints, as the signal-to-noise ratio cannot be amplified arbitrarily. Therefore, additional constraints need to be imposed. One common constraint is the 2-norm constraint on $\tilde{C}^j$, which can be expressed as $\|\tilde{C}^j\|_2 \leq U$, where $U$ is a user-defined upper bound determined based on practical considerations.
\end{remark}}
It is evident that Problem 1 is theoretical in nature and Problem 2 is application-oriented. Both problems pose significant challenges. Problem 1 requires explorations and does not have readily available techniques for analysis. On the other hand, Problem 2 involves an optimization problem with a strong nonlinear constraint making it difficult to cope with.
\subsection{Preliminaries}
Before presenting the main results, the following definitions are introduced.
\begin{definition}\label{def:1}
	For a matrix $M\in \mathbb{R}^{2n\times 2n}$, and $J = \left[\begin{array}{cc}
	0&-I_{n}\\I_{n}&0
	\end{array}\right]\in \mathbb{R}^{2n\times 2n}$, if 
	\begin{equation}
	J^{-1}M^{H}J = M^{-1}.
	\end{equation}
	Then, $M$ is called a symplectic matrix.
\end{definition}

For DARE \eqref{eq:DARE1}, if Assumption \ref{assum:inver} holds, there is a corresponding symplectic matrix $\bar{S}$ that has the following form
\begin{equation}
\bar{S} = \left[\begin{array}{cc}
A^{T}+G_{0}A^{-1}Q & -G_{0}A^{-1}\\
-A^{-1}Q&A^{-1}
\end{array}\right],
\end{equation}
where $G_{0} = \bar{C}^{T}\bar{R}^{-1}\bar{C} \geq 0$.
\begin{definition}
	A symplectic matrix $M \in \operatorname{dom}(Ric)$, if the following conditions are true:
	\begin{enumerate}
		\item[(1)] $M$ has no eigenvalues on the unit circle.
		\item[(2)] The invariant subspace $\alpha_{\_}(M)=Im\left(\begin{array}{c}
		T_{1}\\T_{2}
		\end{array}\right)$ of $M$ associated with stable eigenvalues and $Im\left(\begin{array}{c}
		0\\I
		\end{array}\right)$ are complementary.
	\end{enumerate}
\end{definition}

\begin{lemma}\label{lemma:simplec}(See \cite{zhou1996robust})
	If Assumptions \ref{assum:inver} and \ref{assum:obs} hold, the DARE \eqref{eq:DARE1} has a unique positive-definite solution $\bar{P}$, and the corresponding symplectic matrix $\bar{S} \in dom(Ric)$, i.e., $\bar{S}$ has no eigenvalues on the unit circle, and the invariant subspace $\alpha_{\_}(\bar{S})=Im\left(\begin{array}{c}
	\bar{X}\\\bar{Y}
	\end{array}\right)$ associated with eigenvalues $|z|<1$ and $Im\left(\begin{array}{c}
	0\\I
	\end{array}\right)$ are complementary, and $\bar{P} =\bar{Y}\bar{X}^{-1} $, 
	where $\left(\bar{X}^{T},\bar{Y}^{T}\right)^{T}$ is the matrix composed of the corresponding stable eigenvectors and generalized eigenvectors of $\bar{S}$.
\end{lemma}
Similarly, one can also get that $S \in dom(Ric)$, and $P = YX^{-1}$
\begin{equation}
S= \left[\begin{array}{cc}
A^{T}+GA^{-1}Q & -GA^{-1}\\
-A^{-1}Q&A^{-1}
\end{array}\right],
\end{equation}
where $G = G_{0}+G_{1}$ and $G_{1} = \tilde{C}^{T}\tilde{R}^{-1}\tilde{C}\geq 0$, $\left(X^{T},Y^{T}\right)^{T}$ is the matrix composed of the corresponding stable eigenvectors and generalized eigenvectors of $S$, and $\alpha_{\_}(S)=Im\left(\begin{array}{c}
X\\Y
\end{array}\right)$.
\begin{lemma}\label{lemma:invI+PG}(See \cite{zhou1996robust})
	If $\bar{S}, S \in dom(Ric)$, both $I+\bar{P}G_{0}$ and $I+PG$ are invertible.
\end{lemma}
\begin{lemma}\label{lemma:symmetric}(See \cite{zhou1996robust})
	If $\lambda \in \mathbb{C}$ is an eigenvalue of a symplectic matrix $M$, then $1/{\lambda}$, $\lambda^{\star}$ and $1/\lambda^{\star}$ are also eigenvalues of $M$.
\end{lemma}

\begin{lemma}(Schur complement, see \cite{zhang2006schur})\label{lemma:Schur}
Given the matrices $A_{11}=A_{11}^{T}$, $A_{22}=A_{22}^{T}$, and $A_{12}$ with appropriate dimensions, the following inequalities are equivalent:
	
	1. $\left[\begin{array}{ll}A_{11} & A_{12} \\ A_{12}^{T} & A_{22}\end{array}\right] \leq 0$;
	
	2. $A_{22}=A_{22}^{T} \leq 0 ; A_{11}-A_{12} A_{22}^{-1} A_{12}^{T} \leq 0$;
	
	3. $A_{11}=A_{11}^{T} \leq 0 ; A_{22}-A_{12}^{T} A_{11}^{-1} A_{12} \leq 0$.
\end{lemma}

\begin{lemma}
	\label{lemma:matrix_inv}
	(Woodbury matrix identity). For any matrix $P$, $Q$, $C$ of
	proper dimensions, if $P^{-1}$
	and $Q^{-1}$
	exist, then the following equality
	holds:
	\begin{equation}
	\left(P^{-1}+C^T Q^{-1} C\right)^{-1}=P-P C^T\left(C P C^T+Q\right)^{-1} C P.
	\end{equation}
\end{lemma}

\section{The Effects of Redundant Sensors}
This section aims to investigate the relationship between the error covariance matrices of the original sensor network and the redundant sensor network. Firstly, we establish that the difference between the error covariance matrices of the two networks is positive semidefinite. Moreover, we demonstrate that the steady-state MSE of the redundant sensor network is smaller than that of the original sensor network. Additionally, we provide a sufficient and necessary condition to determine when the difference between the covariance matrices is strictly positive definite. This condition indicates that the newly added redundant sensors effectively enhance the estimation performance of each state element.

\textcolor{black}{From Woodbury matrix identity, the DAREs \eqref{eq:DARE1} and \eqref{eq:DARE2} can be rewritten as follows:}
\begin{equation}\label{eq:DARE11}
A\left(I+\bar{P}G_{0}\right)^{-1}\bar{P}A^{T}-\bar{P}+Q=0,
\end{equation}
\begin{equation}\label{eq:DARE12}
A\left(I+PG\right)^{-1}PA^{T}-P+Q=0,
\end{equation}
where $\bar{P}$ and $P$ are the priori steady-state error covariance matrices of the original sensor network and the redundant sensor network, respectively. The following assumption is made to avoid trivial problems:
\textcolor{black}{
\begin{assumption}\label{assum:3}
	$\tilde{C} \neq 0$, i.e., $G = \tilde{C}^{T}\tilde{R}^{-1}\tilde{C}\neq 0$.
\end{assumption}}

Assumption \ref{assum:3} guarantees that redundant sensors are active, which means that at least one redundant sensor can get the measurement of the process.

The following theorem establishes a significant relationship between $\bar{P}$ and $P$, as well as between $\text{tr}(\bar{P})$ and $\text{tr}(P)$.

\textcolor{black}{
\begin{theorem}\label{thm:1}
	If Assumptions \ref{assum:inver} and \ref{assum:obs} hold, and $\bar{P}$ and $P$ are solutions of the DAREs \eqref{eq:DARE11} and \eqref{eq:DARE12}, representing the priori error covariance matrices obtained by the original sensor network $\bar{C}$ and the sensor network with redundant sensors $C = \left[\bar{C}^{T}~ \tilde{C}^{T}\right]^{T}$, respectively, the following statements hold:
	\begin{enumerate}
		\item[(1)] The covariance matrix of the original sensor network is greater than or equal to that of the redundant sensor network, i.e., $\bar{P}\geq P$.
		\item[(2)] The MSE of the original sensor network is greater than the MSE of the redundant sensor network, i.e., $\operatorname{tr}(\bar{P})> \operatorname{tr}(P)$, if and only if Assumption \ref{assum:3} also holds.
	\end{enumerate}
\end{theorem}
}
	\begin{proof}1): 
		From Assumptions \ref{assum:inver} and \ref{assum:obs}, and $\bar{R}>0$, $\hat{R}>0$, one has that $\bar{P}>0$ and $P>0$. 
		By $G\geq G_{0}$, it can be obtained that 
		\begin{equation}\label{eq:barPP}
			A\left(I+PG_{0}\right)^{-1}PA^{T} \geq 	A\left(I+PG\right)^{-1}PA^{T}.
		\end{equation}
		Define 
		\begin{equation}\label{eq:Fg0}
			\begin{aligned}
				A_{G_{0}} &= A-A\bar{P}\bar{C}^{T}\left(\bar{C}\bar{P}\bar{C}^{T}+\bar{R}\right)^{-1}\bar{C}\\
				&=A\left(I-\bar{P}\bar{C}^{T}\left(\bar{C}\bar{P}\bar{C}^{T}+\bar{R}\right)^{-1}\bar{C}\right)\\
				&=A(I+\bar{P}G_{0})^{-1}.
			\end{aligned}
		\end{equation}
		One has 
		\begin{equation}\label{eq:inFg0}
			\begin{aligned}
				&A\bar{P}\bar{C}^{T}\left(\bar{C}\bar{P}\bar{C}^{T}+R\right)^{-1}\\
				&=A\left(\bar{P}\bar{C}^{T}\bar{R}^{-1}-\bar{P}G_{0}\left(I+\bar{P}G_{0}\right)^{-1}\bar{P}\bar{C}^{T}\bar{R}^{-1}\right)\\
				& = A\left(I+\bar{P}G_{0}\right)^{-1}\bar{P}\bar{C}^{T}\bar{R}^{-1}= A_{G_{0}}\bar{P}\bar{C}^{T}\bar{R}^{-1}.
			\end{aligned}
		\end{equation}
		Substituting \eqref{eq:Fg0} and \eqref{eq:inFg0} into \eqref{eq:DARE1}, one has \eqref{eq:DARE1} can be transformed  as 
		\begin{equation}\label{eq:RicFg0}
			\begin{aligned}
				A_{G_{0}}\bar{P}A_{G_{0}}^{T}-\bar{P}+Q+A_{G_{0}}\bar{P}G_{0}\bar{P}A_{G_{0}}^{T}=0.
			\end{aligned}
		\end{equation}
		From \eqref{eq:RicFg0}, one can get that $A_{G_{0}}$ is Schur stable. 
		Subtracting \eqref{eq:DARE12} from \eqref{eq:RicFg0}, it can be obtained that
		\begin{equation}\begin{aligned}\label{eq:subFg0}
				&A_{G_{0}}\left(\bar{P}-P\right)A_{G_{0}}^{T} -\left(\bar{P}-P\right)+A_{G_{0}}PA_{G_{0}}^{T}\\&+A_{G_{0}}\bar{P}G_{0}\bar{P}A_{G_{0}}^{T}-	A\left(I+PG\right)^{-1}PA^{T} = 0.
			\end{aligned}
		\end{equation}
		By adding and subtracting $A\left(I+PG_{0}\right)^{-1}PA^{T}$ from the left hand side of \eqref{eq:subFg0}, one has 
		\begin{equation}
			\label{eq:subFg01}
			\begin{aligned}
				&A_{G_{0}}\left(\bar{P}-P\right)A_{G_{0}}^{T} -\left(\bar{P}-P\right)+A_{G_{0}}PA_{G_{0}}^{T}\\&+A_{G_{0}}\bar{P}G_{0}\bar{P}A_{G_{0}}^{T}-A\left(I+PG_{0}\right)^{-1}PA^{T}\\&+A\left(I+PG_{0}\right)^{-1}PA^{T}-	A\left(I+PG\right)^{-1}PA^{T} = 0.
			\end{aligned}
		\end{equation}
		Note the fact that  $A_{G_{0}}$ is Schur stable, and 
		\begin{equation}
			\begin{aligned}
				&A_{G_{0}}PA_{G_{0}}^{T}+A_{G_{0}}\bar{P}G_{0}\bar{P}A_{G_{0}}^{T}-A\left(I+PG_{0}\right)^{-1}PA^{T}
				\\=& A_{G_{0}}PA_{G_{0}}^{T}+A_{G_{0}}\bar{P}G_{0}\bar{P}A_{G_{0}}^{T} \\&-  A_{G_{0}}(I+\bar{P}G_{0})(I+PG_{0})^{-1}P(I+\bar{P}G_{0})^{T}A_{G_{0}}^{T}
				\\=&  A_{G_{0}}(\bar{P}-P)G_{0}(I+PG_{0})^{-1}(\bar{P}-P) A_{G_{0}}^{T}\geq 0,
			\end{aligned}
		\end{equation}
		and 
		\begin{equation}
			A\left(I+PG_{0}\right)^{-1}PA^{T} -	A\left(I+PG\right)^{-1}PA^{T}\geq 0.
		\end{equation}
		Therefore, one has the following Lyapunov equation
		\begin{equation}\label{eq:LyaFg0}
			\begin{aligned}
				&A_{G_{0}}\left(\bar{P}-P\right)A_{G_{0}}^{T} -\left(\bar{P}-P\right) \\+&A_{G_{0}}(\bar{P}-P)G_{0}(I+PG_{0})^{-1}(\bar{P}-P) A_{G_{0}}^{T}\\+&A\left(I+PG_{0}\right)^{-1}PA^{T} -	A\left(I+PG\right)^{-1}PA^{T}= 0,
			\end{aligned}
		\end{equation}
		and
		\begin{equation}
			\bar{P}-P\geq0.
		\end{equation}
		
		2):  Since $\bar{P}\geq P$, $\operatorname{tr}(\bar{P}) \geq \operatorname{tr}(P)$ and $\operatorname{tr}(\bar{P}) =\operatorname{tr}(P)$ if and only if $\bar{P}=P$. Therefore,  $\operatorname{tr}(\bar{P}) >\operatorname{tr}(P)$ if and only if $\bar{P}\neq P$.
		
		\textbf{Sufficiency:} if $G1 \neq 0$, and $tr(\bar{P}) =tr(P)$, that is $\bar{P}=P$. Substituting $\bar{P}=P$ into \eqref{eq:LyaFg0}, it can be obtained that 
		\begin{equation}\label{eq:diff=0}
			\begin{aligned}
				A\left(I+PG_{0}\right)^{-1}PA^{T} -	A\left(I+PG\right)^{-1}PA^{T}=0.
			\end{aligned}
		\end{equation}
		Equation \eqref{eq:diff=0} can be transformed as
		\begin{equation}\label{eq:transdiff}
			\begin{aligned}
				&A\left(I+PG_{0}\right)^{-1}PA^{T} -	A\left(I+PG\right)^{-1}PA^{T}
				\\=& A\left(I+PG_{0}\right)^{-1}PA^{T} - A\left(I+PG_{0}+PG_{1}\right)^{-1}PA^{T}
				\\= &A\left(I+PG_{0}\right)^{-1}PA^{T} - A\left(I+PG_{0}\right)^{-1}PA^{T} \\&+A\left(I+PG_{0}\right)^{-1}PG_{1}\left(I + (I+PG_{0})^{-1}PG_{1}\right)^{-1}\\ &\times\left(I+G_{0}P\right)^{-1}A^{T}
				\\=&A\left(I+PG_{0}\right)^{-1}PG_{1}\left(I + (I+PG_{0})^{-1}PG_{1}\right)^{-1}\\ &\times\left(I+G_{0}P\right)^{-1}A^{T}=0.
			\end{aligned}
		\end{equation}
		By using the fact $A$ is invertible and $P>0$, $G1\neq 0$, the equation \eqref{eq:transdiff} can not be true, which leads to a contradiction. Therefore if $G_{1}\neq 0$,  $tr(\bar{P}) >tr(P)$ holds.
		
		\textbf{Necessity:} if $tr(\bar{P}) >tr(P)$, and $G_{1}=0$. Then, equations \eqref{eq:DARE11} and \eqref{eq:DARE12} are the same. Since the DARE \eqref{eq:DARE11} has a unique positive-definite solution, one has $\bar{P} = P$, which leads to a contradiction. Therefore, if  $tr(\bar{P}) >tr(P)$, $G_{1}\neq 0$ holds.	
	\end{proof}

\begin{remark}
	 The proof of result (1) in Theorem \ref{thm:1} establishes the relationship between the steady-state priori error covariance matrices $\bar{P}$ and $P$ by transforming the difference $\bar{P} - P$ into a solution of the Lyapunov equation \eqref{eq:LyaFg0}. This transformation simplifies the analysis and facilitates the proof of result (2) in Theorem \ref{thm:1}. By expressing the difference $\bar{P} - P$ as a solution of the Lyapunov equation \eqref{eq:LyaFg0}, it becomes possible to analyze the properties of the difference matrix and establish the relationships between $\bar{P}$ and $P$. Furthermore, these relationships can directly translate into the relationships between the traces of $\bar{P}$ and $P$.
\end{remark}

Theorem \ref{thm:1} shows that by adding redundant sensors, the sensor network can achieve lower MSE, and  $\bar{P}\geq P$.  However, one may care more about when the strict inequality, i.e., $\bar{P}> P$ holds. The following example illustrates the reason.

\textbf{Example}: Consider three Gaussian random variables $n_{1} \sim N(0,\Sigma_{1})$, $n_{2} \sim N(0,\Sigma_{2})$ and $n_{3} \sim N(0,\Sigma_{3})$ with
\begin{equation}\label{eq:covexample}
\Sigma_{1} = \left[\begin{array}{cc}
1&0\\0&1
\end{array}\right],
\Sigma_{2} = \left[\begin{array}{cc}
1&0\\0&0.5
\end{array}\right],
\Sigma_{3} = \left[\begin{array}{cc}
0.8&0\\0&0.8
\end{array}\right].
\end{equation}
Then, the variables $n_{1}$, $n_{2}$, $n_{3}$ and the corresponding uncertainty ellipsoid are  plotted in \figurename~\ref{fig:plotn}.

\begin{figure}[htbp]
	\centering
	\includegraphics[width=0.7\linewidth]{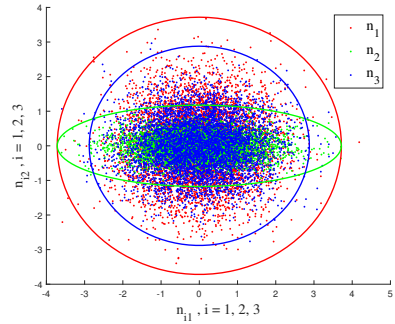}
	\caption{The distributions and corresponding error ellipsoids of $n_{i}$.}
	\label{fig:plotn}
\end{figure}

From \eqref{eq:covexample} and Fig. \ref{fig:plotn}, it is obvious that $\Sigma_{1}\geqq\Sigma_{2}$, $\Sigma_{1}>\Sigma_{3}$ and $tr(\Sigma_{2})<tr(\Sigma_{3})$. However, if one focuses on the first component of the variable, i.e., $n_{i1}$, the variances the $n_{11}$ and $n_{21}$ are equal, and the variance of $n_{31}$ is smaller than that of $n_{11}$. That is, the strict inequality guarantees that the variances of all the components in the random variable are getting smaller, \textcolor{black}{an overall improvement of the estimation performance.} Therefore, it is important to study when $\bar{P}>P$ holds. The following theorem gives the answer.

\begin{theorem}\label{thm:2}
	If Assumptions \ref{assum:inver}, \ref{assum:obs} and \ref{assum:3} hold, and $\bar{P}$ and $P$ are solutions of the DAREs \eqref{eq:DARE11} and \eqref{eq:DARE12} respectively, one has that $\bar{P}> P$ if and only if the corresponding symplectic matrices $\bar{S}$ and $S$ do not have any common stable eigenvalues and associated common eigenvectors simultaneously.
\end{theorem}

	\begin{proof}
		In Theorem \ref{thm:1}, it has been proved that $\bar{P}\geq P$. One has that if $Ker(\bar{P}-P)= {0}$, $\bar{P}>P$. Therefore, proving Theorem \ref{thm:2} is equivalent to prove that $Ker(\bar{P}-P)= {0}$ if and only if $\bar{S}$ and $S$ do not have any common stable eigenvalues and associated common eigenvector simultaneously.

		\textbf{Sufficiency:} If $Ker(\bar{P}-P)\neq {0}$, suppose that $(\bar{P}-P)x = 0$. Considering  \eqref{eq:barPP}, and multiplying left and right sides of \eqref{eq:LyaFg0} by $x^{T}$ and $x$ respectively, one has 
		\begin{equation}\label{eq:Fg0inva}
			x^{T}A_{G_{0}}\left(\bar{P}-P\right)A_{G_{0}}^{T}x = 0,
		\end{equation}
		\begin{equation}\label{eq:diffxt}
			x^{T}A\left(I+PG_{0}\right)^{-1}PA^{T}x -x^{T}	A\left(I+PG\right)^{-1}PA^{T}x = 0.
		\end{equation}
		Equation \eqref{eq:Fg0inva} means $Ker\left(\bar{P}-P\right)$ is $A_{G_{0}}^{T}$ invariant, and there exists an left eigenvalue of $A_{G_{0}}$, $\lambda\neq 0$, and its corresponding eigenvector is $x^{T}$ satisfying $x \in Ker(\bar{P}-P)$, which means $A_{G_{0}}^{T}x = \lambda x$. 
		
		From \eqref{eq:transdiff} and \eqref{eq:diffxt}, one has 
		\begin{equation}\label{eq:transdiffx}
			\begin{aligned}
				&x^{T}A\left(I+PG_{0}\right)^{-1}PG_{1}\left(I + (I+PG_{0})^{-1}PG_{1}\right)^{-1}\\ &\times\left(I+G_{0}P\right)^{-1}A^{T}x=0. 
			\end{aligned}
		\end{equation}
		
		Note that 
		\begin{equation}\label{eq:xGP}
			\begin{aligned}
				\left(I+G_{0}P\right)^{-1}A^{T}x =& \left(I+G_{0}P\right)^{-1}\left(I+G_{0}\bar{P}\right)A_{G_{0}}^{T}x\\
				=& \lambda\left(I+G_{0}P\right)^{-1}\left(I+G_{0}\bar{P}\right)x\\
				=&\lambda\left(I+G_{0}P\right)^{-1}\left(I+G_{0}P\right)x = \lambda x.
			\end{aligned}
		\end{equation}
		From \eqref{eq:transdiffx} and \eqref{eq:xGP}, one can get that 
		\begin{equation}\label{eq:PG1x}
			G_{1}Px = 0.
		\end{equation}
		Let $A_{G} = A\left(I+PG\right)^{-1}$, the following equation holds
		\begin{equation}
			\begin{aligned}
				A_{G}^{T}x &= \left(I+GP\right)^{-1}\left(I+G_{0}\bar{P}\right)A_{G_{0}}^{T}x\\
				&=\lambda \left(I+GP\right)^{-1}\left(\left(I+G_{0}\bar{P}\right)x +	G_{1}Px\right)\\
				&=\lambda x,
			\end{aligned}
		\end{equation}
		which means $A^{T}_{G_{0}}$ and $A^{T}_{G}$ have a common stable eigenvalue and corresponding eigenvector, respectively.
		
		Multiplying $\left(I,\bar{P}\right)^{T}$ to the right side of $\bar{S}$ gives 
		\begin{equation}\label{eq:simplecmulti}
			\begin{aligned}
				\bar{S}\left[\begin{array}{c}
					I\\\bar{P}
				\end{array}\right] &= \left[\begin{array}{cc}
					A^{T}+G_{0}A^{-1}Q & -G_{0}A^{-1}\\
					-A^{-1}Q&A^{-1}
				\end{array}\right]\left[\begin{array}{c}
					I\\\bar{P}
				\end{array}\right] 
				\\&=  \left[\begin{array}{c}
					A^{T}+G_{0}A^{-1}(Q-\bar{P})\\
					-A^{-1}(Q-\bar{P})
				\end{array}\right],
			\end{aligned}
		\end{equation}
		considering the Riccati equation \eqref{eq:DARE11}, one has $(Q-\bar{P}) = -A(I+\bar{P}G_{0})^{-1}\bar{P}A^{T}$. Then,
		\begin{equation}\label{eq:simplecmulticompon}
			\begin{aligned}
				\left[\begin{array}{c}
					A^{T}+G_{0}A^{-1}(Q-\bar{P})\\
					-A^{-1}(Q-\bar{P})
				\end{array}\right]&= \left[\begin{array}{c}
					A^{T}-G_{0}(I+\bar{P}G_{0})^{-1}\bar{P}A^{T}\\
					(I+\bar{P}G_{0})^{-1}\bar{P}A^{T}
				\end{array}\right]\\
				&= \left[\begin{array}{c}
					(I+G_{0}\bar{P})^{-1}A^{T}\\
					\bar{P}(I+G_{0}\bar{P})^{-1}A^{T}
				\end{array}\right]
			\end{aligned}
		\end{equation}
		combining \eqref{eq:simplecmulti} and \eqref{eq:simplecmulticompon}, and noting that $A_{G_{0}} = A(I+\bar{P}G_{0})^{-1}$, it can be obtained that 
		\begin{equation}
			\label{eq:simp1}
			\bar{S}\left[\begin{array}{c}
				I\\\bar{P}
			\end{array}\right] =\left[\begin{array}{c}
				I\\\bar{P}
			\end{array}\right] A_{G_{0}}^{T}.
		\end{equation}
		Similarly, one has
		\begin{equation}
			\begin{aligned}
				S\left[\begin{array}{c}
					I\\{P}
				\end{array}\right] = \left[\begin{array}{c}
					I\\{P}
				\end{array}\right]A_{G}^{T}.
			\end{aligned}
		\end{equation}
		Multiplying the right-hand sides of $x$ and letting $y = \bar{P}x = Px$ yield
		
		\begin{equation}
			\bar{S}\left[\begin{array}{c}
				x\\y
			\end{array}\right] = \lambda\left[\begin{array}{c}
				x\\y
			\end{array}\right],
		\end{equation}
		\begin{equation}
			S\left[\begin{array}{c}
				x\\y
			\end{array}\right] = \lambda\left[\begin{array}{c}
				x\\y
			\end{array}\right].
		\end{equation} 
		Since $\lambda$ is a common stable eigenvalue of $A^{T}_{G_{0}}$ and $A^{T}_{G}$, it can be obtained that $\lambda$ is a common stable eigenvalue of $\bar{S}$ and $S$, and $\left(x^{T},y^{T}\right)^{T}$ is the corresponding eigenvector. 
		
		Therefore, $Ker(\bar{P}-P)= {0}$ if $\bar{S}$ and $S$ do not have any common stable eigenvalues and associated  common  eigenvectors simultaneously.
		
		\textbf{Necessity:} If $\bar{S}$ and $S$ have a common stable eigenvalue $\lambda$ and the associated common eigenvector $\left(x^{T},y^{T}\right)^{T}$ consistent with the blocks of the symplectic matrices.  Since Assumptions \ref{assum:inver} and \ref{assum:obs} hold, from Riccati theory, $\bar{S}$ and $S$ have no eigenvalues on the unit circle, and all the eigenvalues are symmetric about the unit circle. Assume the upper Jordan blocks associated with the stable eigenvalues of $\bar{S}$ and $S$ are $\bar\Lambda$ and $\Lambda$, and the matrices composed of the corresponding eigenvectors and generalized eigenvectors are $\left(\bar{X}^{T},\bar{Y}^{T}\right)^{T}$ and $\left(X^{T},Y^{T}\right)^{T}$, respectively.  Without loss of generality, suppose that the first columns of $\left(\bar{X}^{T},\bar{Y}^{T}\right)^{T}$ and $\left(X^{T},Y^{T}\right)^{T}$ are the same, which is $\left(x^{T},y^{T}\right)^{T}$. Then, one has
		
		\begin{equation}\label{eq:simplec}
			\bar{S}\left[\begin{array}{c}
				\bar{X}\\\bar{Y}
			\end{array}\right] = \left[\begin{array}{c}
				\bar{X}\\\bar{Y}
			\end{array}\right]\bar{\Lambda}, S \left[\begin{array}{c}
				X\\Y
			\end{array}\right] = \left[\begin{array}{c}
				X\\Y
			\end{array}\right]\Lambda.
		\end{equation}
		By using symplectic matrix theory (Lemma \ref{lemma:simplec}) and the Assumptions \ref{assum:inver} and \ref{assum:obs}, one has that 
		\begin{equation}
			\bar{P} = \bar{Y}\bar{X}^{-1}, P = YX^{-1}.
		\end{equation}
		Right multiplying $\bar{P}$ and $P$ by $x$ and considering $x$ is the first column of $\bar{X}$ and $X$, $y$ is the first column of $\bar{Y}$ and $Y$, one has 
		\begin{equation}
			\bar{P}x =  \bar{Y}\bar{X}^{-1}x = y,
		\end{equation}
		\begin{equation}
			Px =  YX^{-1}x = y.
		\end{equation}
		It can be obtained that $\left(\bar{P}-P\right)x=0$, which means  $Ker(\bar{P}-P)\neq {0}$. Therefore, $Ker(\bar{P}-P)= {0}$ only if $\bar{S}$ and $S$ do not have any common stable eigenvalues and associated common eigenvectors simultaneously.
	\end{proof}

\begin{remark}
 Theorem \ref{thm:2} provides a fundamental condition to determine when the addition of redundant sensors leads to an improvement in the estimation accuracy of all elements of the system state $x_k$ simultaneously. The condition $\bar{P} > P$ indicates that the steady-state priori error covariance matrix $\bar{P}$ is strictly greater than the priori error covariance matrix $P$. This implies that the estimation error variances of all elements in the system state decrease when redundant sensors are added. By satisfying the condition $\bar{P} > P$, the estimation performance of each element of the system state can be enhanced, leading to an overall improvement in the estimation accuracy of the entire system state. This result is significant as it provides a criterion for determining whether the addition of redundant sensors can lead to a simultaneous improvement in the estimation accuracy of all channels.
\end{remark}
\begin{corollary}\label{coro:1}
	If Assumptions \ref{assum:inver}, \ref{assum:obs} and \ref{assum:3} hold, and $\bar{P}$ and $P$ are solutions of the DAREs \eqref{eq:DARE11} and \eqref{eq:DARE12} respectively, one has that $\bar{P}> P$ if and only if the corresponding symplectic matrices $\bar{S}$ and $S$ do not have any common unstable left eigenvalues and associated common left eigenvectors simultaneously.
\end{corollary}
\begin{proof}
	From the definition of symplectic matrix (Definition \ref{def:1}) and Lemma \ref{lemma:symmetric} and Theorem \ref{thm:2}, Corollary \ref{coro:1} can be easily proved.
\end{proof}

Indeed, the analysis and results presented so far have provided insights and answers to Problem \ref{prob:1}, which aims to understand the effects of redundant sensors on the collaborative state estimation performance of a sensor network. Theorems \ref{thm:1} and \ref{thm:2} establish the relationship between the error covariance matrices and provide conditions under which the estimation performance can be enhanced by adding redundant sensors. These results contribute to the understanding of the effects of redundant sensors and provide theoretical insights into the optimization of sensor network design. Therefore, with the analysis conducted so far, significant progress has been made in addressing Problem \ref{prob:1} and gaining insights into the effects of redundant sensors on estimation performance.

\section{The Optimal Design of the Redundant Sensors}
The design of the sensors influences the sensors' performance, as it is proved that the redundant sensors can make the estimate MSE smaller, one may consider the problem that given $N$ sensors, sensor $i$, $i = 1,\ldots,N$, how to design the $m$ redundant sensors, sensor $i$, $i = N+1,\ldots, N+m$, to achieve a better performance. Before proposing the algorithm on how to optimally design the redundant sensors, the following useful lemmas are introduced.

\begin{lemma}\label{lemma:RicIneq}
	If Assumptions \ref{assum:inver} and \ref{assum:obs} hold, and $P$ is solution of the DARE \eqref{eq:DARE12}, for any solution $P^{i}$ to the following discrete-time algebraic Riccati inequality
	\begin{equation}\label{eq:DARI}
	A\left(I+P^{i}G\right)^{-1}P^{i}A^{T}-P^{i}+Q\leq 0,
	\end{equation}
	it follows that 
	\begin{equation}
	P\leq P^{i}.
	\end{equation}
\end{lemma}

Rewrite  \eqref{eq:DARI} as in following form by invoking Lemma \ref{lemma:Schur}
\begin{equation}\label{eq:MDARI}
\left[\begin{array}{cc}
-P^{i}+Q&AP^{i}\\P^{i}A^{T}&-\left(P^{i}+P^{i}GP^{i}\right)
\end{array}\right]\leq 0.
\end{equation}
By multiplying $(P^{i})^{-1}$ to the right and left side of all blocks in \eqref{eq:MDARI}, one has \eqref{eq:MDARI} is equivalent to
\begin{equation}\label{eq:LMItrans}
\left[\begin{array}{cc}
-X+XQX& XA\\A^{T}X& -X-G
\end{array}\right]\leq 0 ,
\end{equation}
where $X = (P^{i})^{-1}$. 

From the mathematical operations described above, it can be observed that the Riccati inequality \eqref{eq:DARI} is equivalent to \eqref{eq:LMItrans}. To facilitate the development of the optimal design algorithm, an auxiliary matrix $C_{r}$ is introduced through the following lemma:

\begin{lemma}\label{lemma:ineqTrans}
	The inequality \eqref{eq:LMItrans} holds if and only if there exists a \textcolor{black}{$C_{r} \in \mathbb{R}^{\tilde{m}\times n}$} such that 
	\begin{equation}\label{eq:ineqtrans}
	\left[\begin{array}{cc}
	-X+XQX& XA\\A^{T}X& \Omega
	\end{array}
	\right]\leq 0,
	\end{equation}
	where \begin{equation}\label{eq:Omega}
	\Omega = -\left(X+G_{0}\right) - \tilde{C}^{T}\tilde{R}^{-1}\tilde{C}+\left(\tilde{C}-C_{r}\right)^{T}\tilde{R}^{-1}\left(\tilde{C}-C_{r}\right).
	\end{equation}
\end{lemma}

\begin{proof}
	\textbf{Sufficiency:} If there exists a matrix $C_{r}$ such that \eqref{eq:ineqtrans} holds. From Lemma \ref{lemma:Schur}, one has 
	\begin{equation}\label{eq:-x+xq<0}
	-X+XQX\leq 0,
	\end{equation}
	and 
	\begin{equation}
	\Omega-A^{T}X(-X+XQ)^{-1}XA\leq 0.
	\end{equation}
	By \eqref{eq:Omega}, it follows that 
	\begin{equation}
	-X-G-A^{T}X(-X+XQ)^{-1}XA\leq-\Upsilon,
	\end{equation}
	where $\Upsilon = \left(\bar{C}-C_{r}\right)^{T}\hat{R}^{-1}\left(\bar{C}-C_{r}\right)\geq0$. Then, one has 
	\begin{equation}\label{eq:-x-g<0}
	-X-G-A^{T}X(-X+XQ)^{-1}XA \leq 0.
	\end{equation}
	From \eqref{eq:-x+xq<0}, \eqref{eq:-x-g<0} and Lemma \ref{lemma:Schur}, one has \eqref{eq:LMItrans} holds.
	
	\textbf{Necessity:} If the inequality \eqref{eq:LMItrans} holds. Then, set $C_{r} =\hat{C}$,  one has that \eqref{eq:ineqtrans} holds.
\end{proof}

The following theorem is presented to guide the design optimization algorithm for redundant sensors. It introduces a parameter $\gamma$ to establish the relationship between the traces of covariance matrices and utilizes the auxiliary variable $C_r$ introduced in Lemma \ref{lemma:ineqTrans} to facilitate the iteration.

\begin{theorem}\label{thm:3}
	Suppose that Assumptions \ref{assum:inver}, \ref{assum:obs} and \ref{assum:3} hold, and $\bar{P}$ is the unique solution to the Riccati equation \eqref{eq:DARE1}. For a given matrix $C_{r}$, if there exists $X=X^{T}>0$ and $\tilde{C}$ such that
\textcolor{black}{	\begin{equation}F(X,\tilde{C}) = \label{eq:LMImain}
	\left[\begin{array}{cccc}
	-X& XA&X\sqrt{Q}&0\\
	A^{T}X&\Phi&0&C_{r}^{T}\\
	\sqrt{Q}X&0&-I&0\\
	0&C_{r}&0&-\tilde{R}
	\end{array}\right]\leq 0,
	\end{equation}}
	where $\Phi = -X-G_{0}-\tilde{C}^{T}\tilde{R}^{-1}C_{r}-C_{r}^{T}\tilde{R}^{-1}\tilde{C}$, and 
\textcolor{black}{	\begin{equation}\label{eq:LMItrace}
	\left[\begin{array}{cc}
	\gamma &\left[e^{T}_{1,n},\ldots, e^{T}_{n,n}\right]\\ \left[\begin{array}{c}
	e_{1,n}\\ \vdots\\e_{n,n}
	\end{array}\right]& I_{n} \otimes X
	\end{array}\right]\geq 0,
	\end{equation}}
	then, the traces of the covariances $\bar{P}$ and $P$ satisfy 
	\begin{equation}
	\operatorname{tr}(\bar{P}) - \operatorname{tr}(P) \geq \operatorname{tr}(\bar{P})-\gamma.
	\end{equation}
\end{theorem}

	\begin{proof}
		By invoking Lemma \ref{lemma:Schur}, \eqref{eq:LMImain} is equivalent to 
		\begin{equation}\label{eq:LMImainSchur}
			\begin{aligned}
				\left[\begin{array}{cc}
					-X&XA\\A^{T}X&\Phi
				\end{array}\right]&-
				\left[\begin{array}{cc}
					X\sqrt{Q}&0\\0&C_{r}^{T}	
				\end{array}\right]
				\left[\begin{array}{cc}
					-I&0\\0&-\hat{R}
				\end{array}\right]^{-1}\\&\times\left[\begin{array}{cc}
					\sqrt{Q}X&0\\0&C_{r}	
				\end{array}\right]\\
				& = \left[\begin{array}{cc}
					-X+XQX& XA\\A^{T}X& \Phi+C_{r}^{T}\hat{R}^{-1}C_{r}
				\end{array}\right]\leq 0 .
			\end{aligned}
		\end{equation}
		Since 
		\begin{equation}\label{eq:form_square}
			\begin{aligned}
				\Phi+C_{r}^{T}\hat{R}^{-1}C_{r}&=	-\left(X+G_{0}\right) - \hat{C}^{T}\hat{R}^{-1}\hat{C}\\&+\left(\hat{C}-C_{r}\right)^{T}\hat{R}^{-1}\left(\hat{C}-C_{r}\right) = \Omega,
			\end{aligned}
		\end{equation}
		substitute \eqref{eq:form_square} into \eqref{eq:LMImainSchur},	it follows that \eqref{eq:LMImain} is equivalent to \eqref{eq:ineqtrans}.
		
		From Lemma  \ref{lemma:RicIneq} and \ref{lemma:ineqTrans}, one has that $P\leq X^{-1}$ and this means $\operatorname{tr}(P)\leq \operatorname{tr}(X^{-1})$.
		
		Then, by invoking Lemma \ref{lemma:Schur}, \eqref{eq:LMItrace} is equivalent to 
		\begin{equation}
			\gamma-\left[e^{T}_{1},\ldots, e^{T}_{n}\right]\left(I_{n}\otimes X\right)^{-1}\left[e^{T}_{1},\ldots, e^{T}_{n}\right]^{T}\geq 0,
		\end{equation} 
	Note that 
	\begin{equation}
		\left[e^{T}_{1},\ldots, e^{T}_{n}\right] (I_{n} \otimes X^{j})^{-1} \left[\begin{array}{c}
			e_{1}\\ \vdots\\e_{n}
		\end{array}\right]  = tr((X^{j})^{-1}),
	\end{equation}
		which means 
		\begin{equation}
			\gamma\geq \operatorname{tr}(X^{-1})\geq \operatorname{tr}(P).
		\end{equation}
		Therefore, it can be obtained that 
		\begin{equation}
			\operatorname{tr}(\bar{P}) - \operatorname{tr}(P) \geq \operatorname{tr}(\bar{P})-\gamma.
		\end{equation}	
	\end{proof}

\begin{remark}
Theorem \ref{thm:3} is remarkable as it establishes a concise and quantitative relationship between the error covariance traces of the original sensor network and the sensor network with redundant sensors only using a parameter  $\gamma$. Therefore, if the $\gamma$ is minimized, the lower bound of the estimation performance gap (the difference of MSEs) between the original sensor network and the redundant sensor network is maximized. This will lead to a systematical sensor design  strategy (Algorithm 1).
\end{remark}

\begin{algorithm}[h]
	\caption{Optimal Design of Redundant Sensors} 
	\hspace*{0.02in} {\bf Initialization:} 
	 $ C_{r,0}$, $\epsilon>0$,  $j=0$.\\
	\hspace*{0.02in} {\bf Repeat:} 
	\begin{algorithmic}[1]
		\State Set $C_{r} = C_{r,j}$, solve the following semi-definite programming problem
		\begin{equation}\label{eq:opt}
		\{X^{j},\tilde{C}^{j}\} = \arg\min_{X,\tilde{C}}  \gamma^{j},
		\end{equation}
	\textcolor{black}{subject to $F(X,\tilde{C})\leq 0$, $X^{j} = (X^{j})^{T}>0$ and }	
		\textcolor{black}{\begin{equation}\label{eq:LMItracej}
		\left[\begin{array}{cc}
		\gamma^{j} &\left[e^{T}_{1,n},\ldots, e^{T}_{n,n}\right]\\ \left[\begin{array}{c}
		e_{1,n}\\ \vdots\\e_{n,n}
		\end{array}\right]& I_{n} \otimes X^{j}
		\end{array}\right]\geq 0.
		\end{equation}}
		\If{$|\gamma^{j}-\gamma^{j-1}|<\epsilon$ }
		\State  Set $\tilde{C}^{*} = \tilde{C}^{j}$, $X^{*}=X^{j}$, $\gamma^{*} = \gamma^{j}$, and stop.
		\Else
		\State Set $j = j+1$, $C_{r,j} = \tilde{C}^{j}$, and go to step 1.
		\EndIf
		
	\end{algorithmic}
	\hspace*{0.02in} {\bf Output:} 
	output $\tilde{C}^{*}$.
\end{algorithm}

\textcolor{black}{
\begin{theorem}\label{thm:4} 
	Under the conditions of Assumptions 1, 2, and 3, the application of Algorithm 1 ensures the convergence of $\gamma^j$ to $\gamma^{*}$. The resulting output $\tilde{C}^{*}$ represents an optimal configuration of redundant sensors added in the network, and the associated error covariance $P$ of the network satisfies the inequality:
	\begin{equation}\label{eq:gamma*ineq}
	\operatorname{tr}(\bar{P}) - \operatorname{tr}(P) \geq \operatorname{tr}(\bar{P}) - \gamma^{*},
	\end{equation}
	where $\bar{P}$ and $P$ are the error covariances of the filter without and with redundant sensors, respectively.
\end{theorem}
}
	The proof is given in Appendix \ref{proof:thm4}.
	\begin{proof}\label{proof:thm4}
		By Theorem \ref{thm:3}, it can be obtained that \eqref{eq:gamma*ineq} holds. Therefore, only the convergence of Algorithm 1 needs to be proved. First, it is proved that $\gamma^{j}$ has a lower bound. From \eqref{eq:LMItrace}, one has 
		\begin{equation}
			\gamma^{j} \geq tr(X^{-1}).
		\end{equation}
		By Lemma \ref{lemma:RicIneq} and \ref{lemma:ineqTrans}, it can be obtained that 
		\begin{equation}
			X^{-1}\geq P.
		\end{equation}
		Note that DARE \eqref{eq:DARE2} can be rewritten as 
		\begin{equation}
			P = A_{G}PA_{G}^{T}+Q+ A_{G}PGPA_{G}^{T},
		\end{equation}
		where $A_{G} = A(I+PG)^{-1}$. It can be obtained that for $\forall$$j$,
		\begin{equation}
			\gamma^{j}\geq tr(X^{-1})\geq tr(P) \geq tr(Q).
		\end{equation}
		Then, it will be proved that $\gamma^{j+1}\leq\gamma^{j}$. By invoking Lemma \ref{lemma:Schur}, one has \eqref{eq:LMImain} is equivalent to \eqref{eq:ineqtrans}. That is to say, in iteration step $j$ of Algorithm 1, the following holds
		\begin{equation}\label{eq:ineqtransj}
			\left[\begin{array}{cc}
				-X^{j}+X^{j}QX^{j}& X^{j}A\\A^{T}X^{j}& \Omega_{j}
			\end{array}
			\right]\leq 0,
		\end{equation}
		where 
		\begin{equation}\begin{aligned}
				\Omega_{j} = &\left(\hat{C}^{j}-C_{r,j}\right)^{T}\hat{R}^{-1}\left(\hat{C}^{j}-C_{r,j}\right)\\
				&-\left(X^{j}+G_{0}\right) - (\hat{C}^{j})^{T}\hat{R}^{-1}\hat{C}^{j},
			\end{aligned}
		\end{equation} and 
		\begin{equation}\label{eq:LMItracejj}
			\left[\begin{array}{cc}
				\gamma^{j} &\left[e^{T}_{1},\ldots, e^{T}_{n}\right]\\ \left[\begin{array}{c}
					e_{1}\\ \vdots\\e_{n}
				\end{array}\right]& I_{n} \otimes X^{j}
			\end{array}\right]\geq 0.
		\end{equation}
		Then, at iteration step $j+1$, $C_{r}$ in \eqref{eq:ineqtrans} has been set as  $C_{r} = C_{r,j+1} = \hat{C}^{j}$. Still substitute the last step solution  $\{X^{j},\hat{C}^{j}\}$ into the constraint \eqref{eq:LMImain}. Then, by invoking Lemma \ref{lemma:Schur}, one can get that the constraint \eqref{eq:LMImain} becomes 
		\begin{equation}\label{eq:ineqtransj+1}
			\left[\begin{array}{cc}
				-X^{j}+X^{j}QX^{j}& X^{j}A\\A^{T}X^{j}& \Upsilon_{j}
			\end{array}
			\right]\leq 0,
		\end{equation}
		where 
		\begin{equation}\begin{aligned}\label{eq:up<om}
				\Upsilon_{j} =-\left(X^{j}+G_{0}\right) - (\hat{C}^{j})^{T}\hat{R}^{-1}\hat{C}^{j}\leq\Omega_{j},
			\end{aligned}
		\end{equation}
		From \eqref{eq:ineqtransj} and \eqref{eq:up<om}, one has that \eqref{eq:ineqtransj+1} holds. Obviously, \eqref{eq:LMItracejj} also holds in the $j+1$ step for $X^{j}$, $\hat{C}^{j}$, which means  $\{X^{j},\hat{C}^{j}\}$ is a feasible solution for \eqref{eq:opt} in Algorithm 1 at iteration step $j+1$, and $\gamma^{j}$ is a sub-optimal feasible value. Therefore, one has $\gamma^{j+1}\leq\gamma^{j}$.
		
		Since $\gamma^{j}$ is monotone decreasing and has a lower bound, $\gamma^{j}$ will converge, and this completes the proof.
	\end{proof}

While Theorem \ref{thm:4} guarantees the convergence of Algorithm 1
, it does not provide a direct proof of solving Problem 2. The following Theorem is provided to demonstrate that Algorithm 1 exactly solves Problem 2.

\textcolor{black}{
\begin{theorem}
	\label{thm:5}
	If Assumptions 1, 2 and 3 hold, the optimal configuration $\tilde{C}^{*}$ obtained by Algorithm 1 solves Problem \ref{prob:2}, i.e.,
	\begin{equation}
			\tilde{C}^{*} = \arg\min_{\tilde{C}} tr(P), ~ s.t. ~ \eqref{eq:p2co}.
	\end{equation}
\end{theorem}}
\begin{proof}
	First, it is proved that optimizing $\gamma^{j}$ is equivalent to optimizing $tr((X^{j})^{-1})$. From Algorithm 1, it can be obtained that $\gamma^{j}$ is only constrained by \eqref{eq:LMItracej}. 
	
	By invoking Lemma \ref{lemma:Schur}, one has  
	\begin{equation}
		\gamma^{j}\geq 0,
	\end{equation}
	and
	\begin{equation}
		\gamma^{j}-tr((X^{j})^{-1})\geq 0,
	\end{equation}
	which indicates that the optimal $\gamma^{j}$ at each iteration step in Algorithm 1 equals to $tr((X^{j})^{-1})$.
	
	Then, if Algorithm 1 converges at step $j$, it is proved that 
	\begin{equation}
		\hat{C}^{*} = C_{r,j}.
	\end{equation}
	If $\hat{C}^{*} \neq C_{r,j}$, similar to the proof of Theorem \ref{proof:thm4},  set $C_{r,j} =\hat{C}^{*}$, one has that $\{C^{*},X^{*}\}$ is just a feasible solution but may not optimal to the optimization problem \eqref{eq:opt} in iteration step $j+1$, this leads to contradiction. Therefore, $\hat{C}^{*} = C_{r,j}$ is proved.
	
	Finally, it is proved that $P = (X^{*})^{-1}$ is the solution of the DARE \eqref{eq:DARE2}. By invoking Lemma \ref{lemma:Schur} and considering $\hat{C}^{*} = C_{r,j}$, it can be obtained that the constraint \eqref{eq:LMImain} of $X^{*}$ can be converted to 
	\begin{equation}\label{eq:LMItrans*}
		\left[\begin{array}{cc}
			-X^{*}+X^{*}QX^{*}& X^{*}A\\A^{T}X^{*}& -X^{*}-G
		\end{array}\right]\leq 0.
	\end{equation}
	By multiplying $(X^{*})^{-1}$ to the right and left side of all blocks in \eqref{eq:LMItrans*} and denote $(X^{*})^{-1}$ as $P$, it follows
	\begin{equation}\label{eq:MDARIP}
		\left[\begin{array}{cc}
			-P +Q&AP \\P A^{T}&-\left(P +P GP \right)
		\end{array}\right]\leq 0,
	\end{equation}
	
	and $\min tr((X^{*})^{-1})$ becomes $\min tr(P)$, and $P$ is constrained by \eqref{eq:MDARIP}, which is equivalent to 
	\begin{equation}
		-P +Q\leq 0,
	\end{equation}
	and 
	\begin{equation}
		A\left(I+PG\right)^{-1}PA^{T}-P+Q\leq 0.
	\end{equation}
	
	From Lemma \ref{lemma:RicIneq}, one can obtain that $tr(P)$ is optimized when 
	\begin{equation}
		A\left(I+PG\right)^{-1}PA^{T}-P+Q= 0,
	\end{equation}
	which is another form of the DARE \eqref{eq:DARE2}. Therefore, implementing Algorithm 1 solves Problem \ref{prob:2}.
\end{proof}

Theorem \ref{thm:4} and \ref{thm:5} demonstrate the convergence and optimality of Algorithm 1.  Until now, Problem \ref{prob:2} has been successfully solved. The proposed Algorithm 1 guarantees the convergence to an optimal solution that maximizes the lower bound of the performance gap between the original sensor network and the redundant sensor network. Furthermore, Theorem \ref{thm:5} demonstrates that the solution obtained by Algorithm 1 exactly solves Problem \ref{prob:2}, where the goal is to minimize the trace of the covariance matrix while satisfying the constraints. Overall, the paper provides a systematic sensor design strategy and offers insights into optimizing the configuration of redundant sensors to improve estimation performance.

\section{Simulation Examples}
In this section, we provide simulation examples to demonstrate the validity of our results. Specifically, we consider a scenario where four sensors collaborate to estimate the system state \cite{Distributed4SteadyKalmanFilter}:

\begin{equation}\label{eq:sim1}
A = \left[\begin{array}{cc}
0.9 & 0 \\
0 & 1.1
\end{array}\right], \bar{C} = \left[\begin{array}{cccc}
1 & 0 & 1 & 1 \\
0 & 1 & 1 & -1
\end{array}\right]^{T},Q = I_{2}/4, \bar{R} = I_{4}.
\end{equation}
To evaluate the impact of redundant sensors, we introduce two additional sensors configured as 
\begin{equation}
\tilde{C}_{1} = \left[3,0\right], \tilde{R}_{1} = I_{1}, \tilde{C}_{2} = \left[3,3\right], \tilde{R}_{2} = I_{1},
\end{equation}
respectively.

\subsection{The effects of redundant sensors}

In this case, we consider two sensor networks with different redundant sensors. The first network consists of the original sensor network and two additional $\tilde{C}_{1}$ sensors, and its parameters are as follows:
\textcolor{black}{
\begin{equation}
	C_{r1} = \left[\begin{array}{ccc}
	\bar{C}^{T}&\tilde{C}^{T}_{1}&\tilde{C}^{T}_{1}
	\end{array}\right]^{T}, ~R_{r1} = I_{6},
\end{equation}}
the second network consists of the original sensor network and $\tilde{C}_{1}$, $\tilde{C}_{2}$, and its parameters are as follow:\textcolor{black}{
\begin{equation}\label{eq:NRS12}
	C_{r2} = \left[\begin{array}{ccc}
		\bar{C}^{T}&\tilde{C}^{T}_{1}&\tilde{C}^{T}_{2}
	\end{array}\right]^{T}, ~R_{r2} = I_{6}.
\end{equation}}
\textcolor{black}{
The two networks are denoted as $C_{r1}$ and $	C_{r2}$ respectively. Denote the priori error covariance matrices of the  $\bar{C}$, $C_{r1}$, and $C_{r2}$ as $\bar{P}$, $P_{r1}$ and $P_{r2}$, respectively.  The posteriori error covariance matrices of $\bar{C}$, $C_{r1}$, and $C_{r2}$ are denoted as $\bar{P}_{p}$, $P_{p,r1}$ and $P_{p,r2}$, respectively. Also, the corresponding symplectic matrices are denoted as $\bar{S}$, $S_{r1}$ and $S_{r2}$, and the corresponding stable eigenvalues and eigenvectors are denoted as $\bar{\lambda}_{i}$, $\bar{v}_{i}$,  $\lambda_{r1,i}$, $v_{r1,i}$, $\lambda_{r2,i}$,  $v_{r2,i}$, $i = 1, 2$.}

\textcolor{black}{By computation, it is evident that $\bar{\lambda}_{2} = \lambda_{r1,2}$ and $\bar{v}_{2}=v_{r1,2}$. This implies that $\bar{S}$ and $S_{r1}$ share a common stable eigenvalue and its associated eigenvector, while $\bar{S}$ and $S_{r2}$ do not have any common stable eigenvalues and associated eigenvectors.}  By invoking Theorem \ref{thm:2} and Lemma \ref{lemma:congruence}, one has that 
\textcolor{black}{
\begin{equation*}\label{eq:P>=}
\bar{P}\geqq P_{r1},   ~\bar{P}_{p}\geqq P_{p,r1},\bar{P}> P_{r2},   ~\bar{P}_{p}> P_{p,r2}.
\end{equation*}}
 
%
\textcolor{black}{Perform Kalman filtering 20000 steps using $\bar{C}$, $C_{r1}$, and $C_{r2}$,  and the corresponding estimated errors of $\bar{C}$, $C_{r1}$ and $C_{r2}$ as $\bar{e}$, $e_{r1}$ and $e_{r2}$, respectively. We plot \figurename ~\ref{fig:e1} and \figurename~ \ref{fig:e2} to illustrate the performance improvements of $C_{r1}$ and $C_{r2}$ with respect to $\bar{C}$. }

\textcolor{black}{Fig. \ref{fig:e1} illustrates the probability density distribution histograms of $\bar{e}_{1}$, $e_{r1,1}$ (\figurename~ \ref{Fig.sub:e11}) and $\bar{e}_{2}$, $e_{r1,2}$ (\figurename~ \ref{Fig.sub:e12}), one can obtain that by adding two $\tilde{C}_{1}$ sensors, the sensor network $C_{r1}$ reduces the variance of the estimation error of the first element, but fail to reduce the variance of the estimation error of the second element. Fig. \ref{fig:e2} shows the probability density distribution histograms of  $\bar{e}_{1}$, $e_{r2,1}$ (\figurename ~\ref{Fig.sub:e21}) and ,  $\bar{e}_{2}$, $e_{r2,2}$ (\figurename ~\ref{Fig.sub:e22}), it can be obtained that sensor network $C_{r2}$ reduces the variances of both elements' estimate error.} The results validate the correctness of Theorem \ref{thm:2}.

\begin{figure}[htbp]
	\centering
	\subfigure[Probability distributions of the $\bar{e}_{1}$ and  $e_{r1,1}$.]{
		\label{Fig.sub:e11}
		\includegraphics[width=0.23\textwidth]{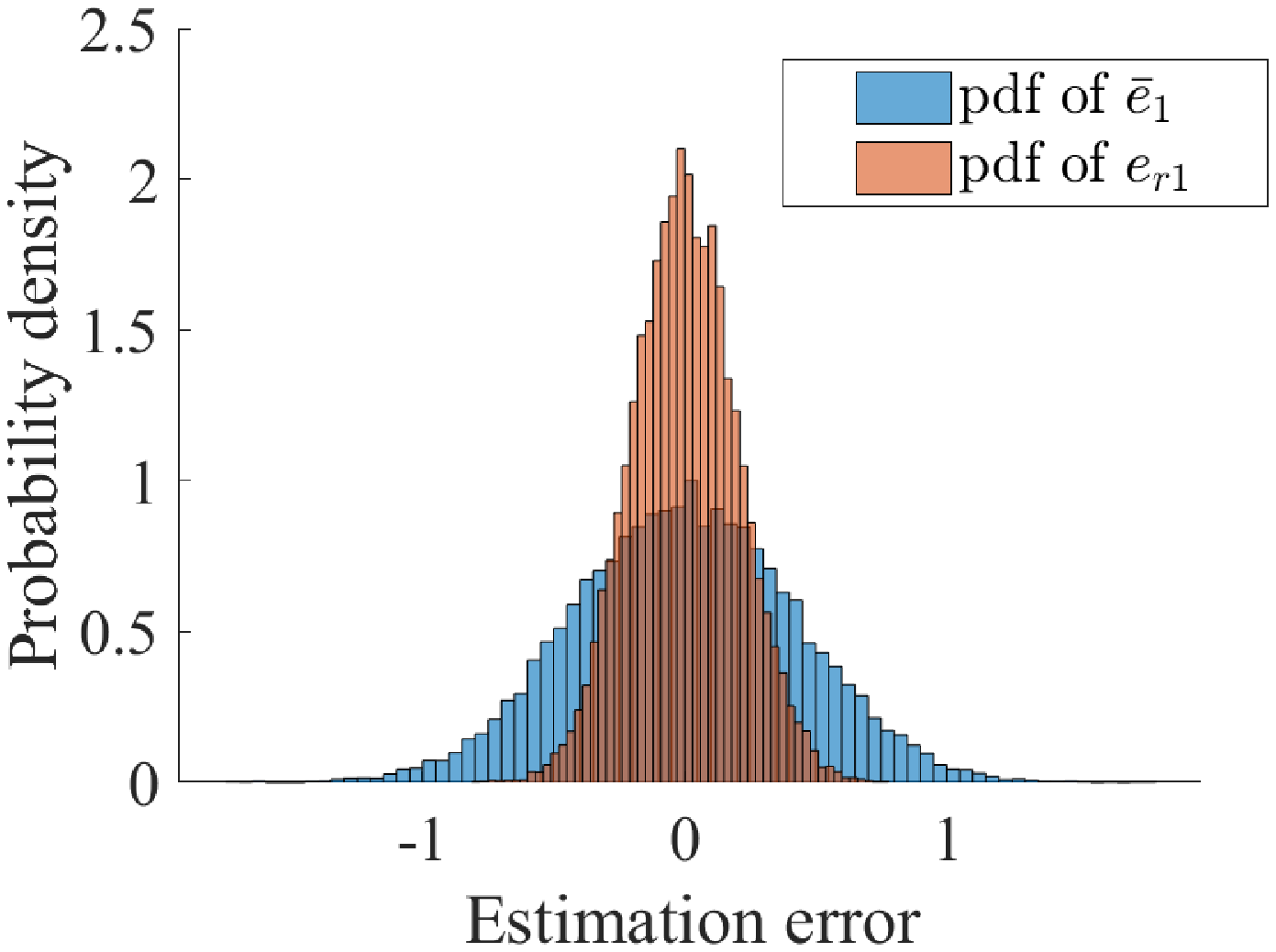}}
	\subfigure[Probability distributions of the $\bar{e}_{2}$ and  $e_{r1,2}$.]{
		\label{Fig.sub:e12}
		\includegraphics[width=0.23\textwidth]{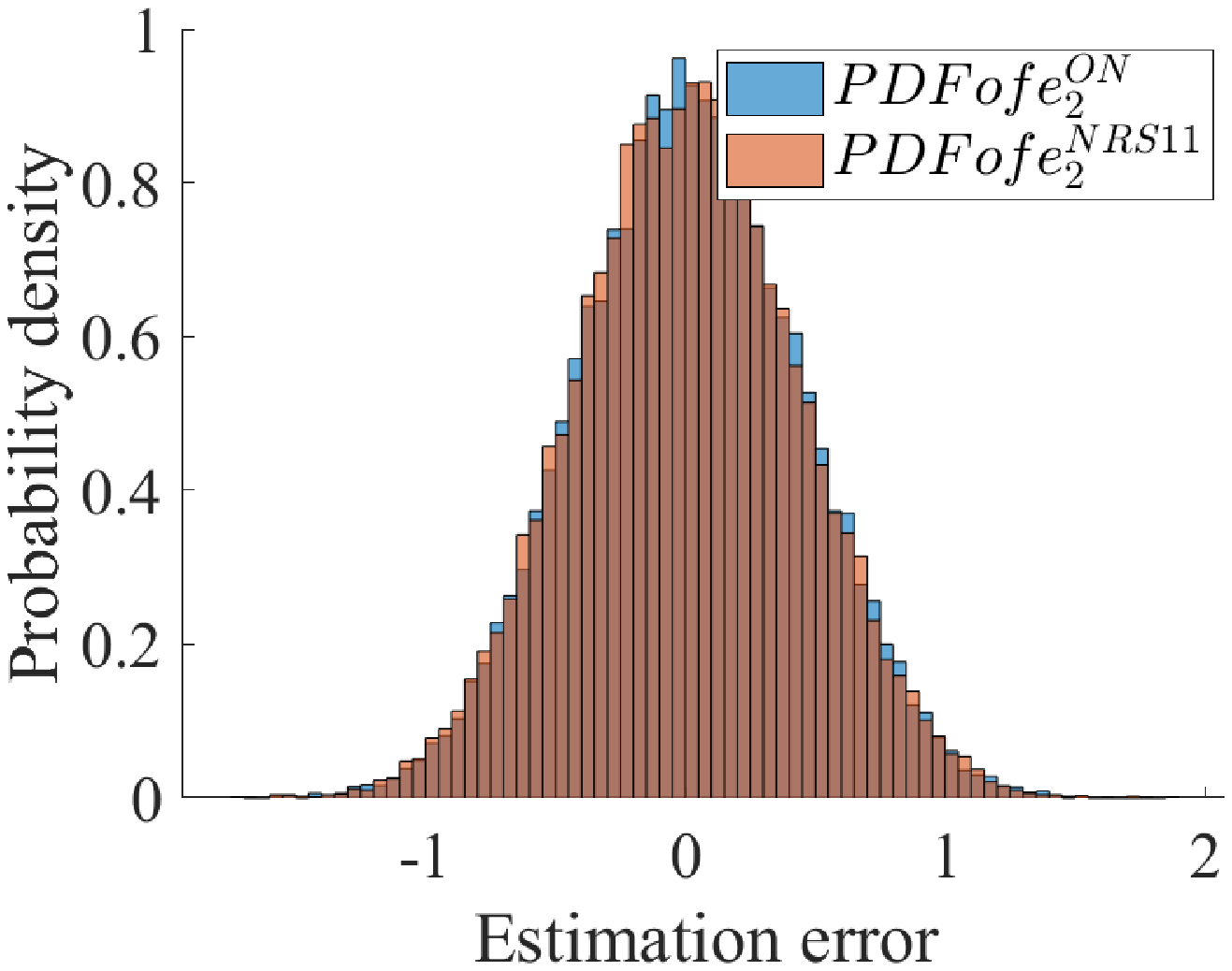}}
	\caption{Probability distributions of the estimate errors of $\bar{C}$ and $C_{r1}$.}
	\label{fig:e1}
\end{figure}
\begin{figure}[htbp]
	\centering
	\subfigure[Probability distributions of the $\bar{e}_{1}$ and  $e_{r2,1}$.]{
		\label{Fig.sub:e21}
		\includegraphics[width=0.23\textwidth]{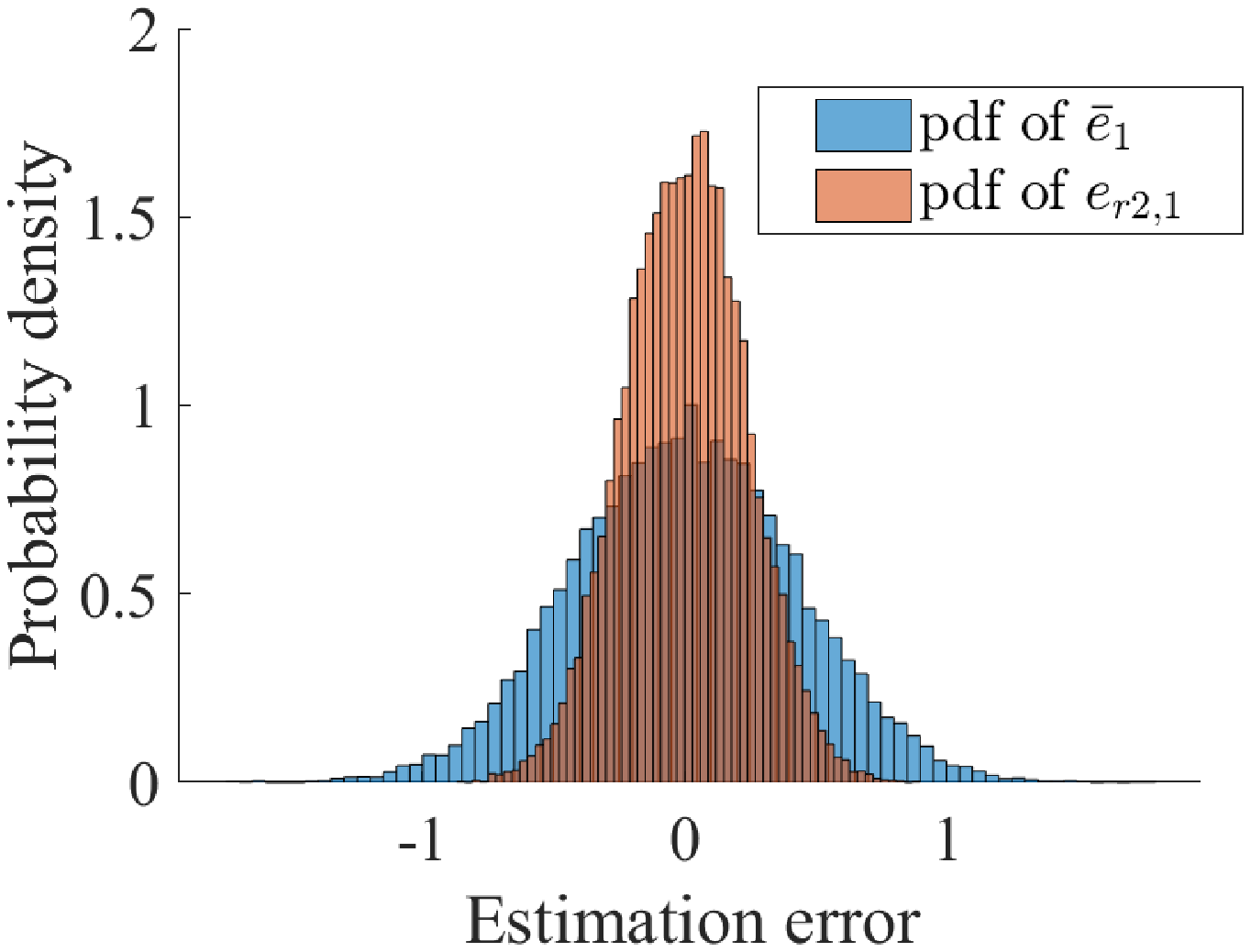}}
	\subfigure[Probability distributions of the $\bar{e}_{2}$ and  $e_{r2,2}$.]{
		\label{Fig.sub:e22}
		\includegraphics[width=0.23\textwidth]{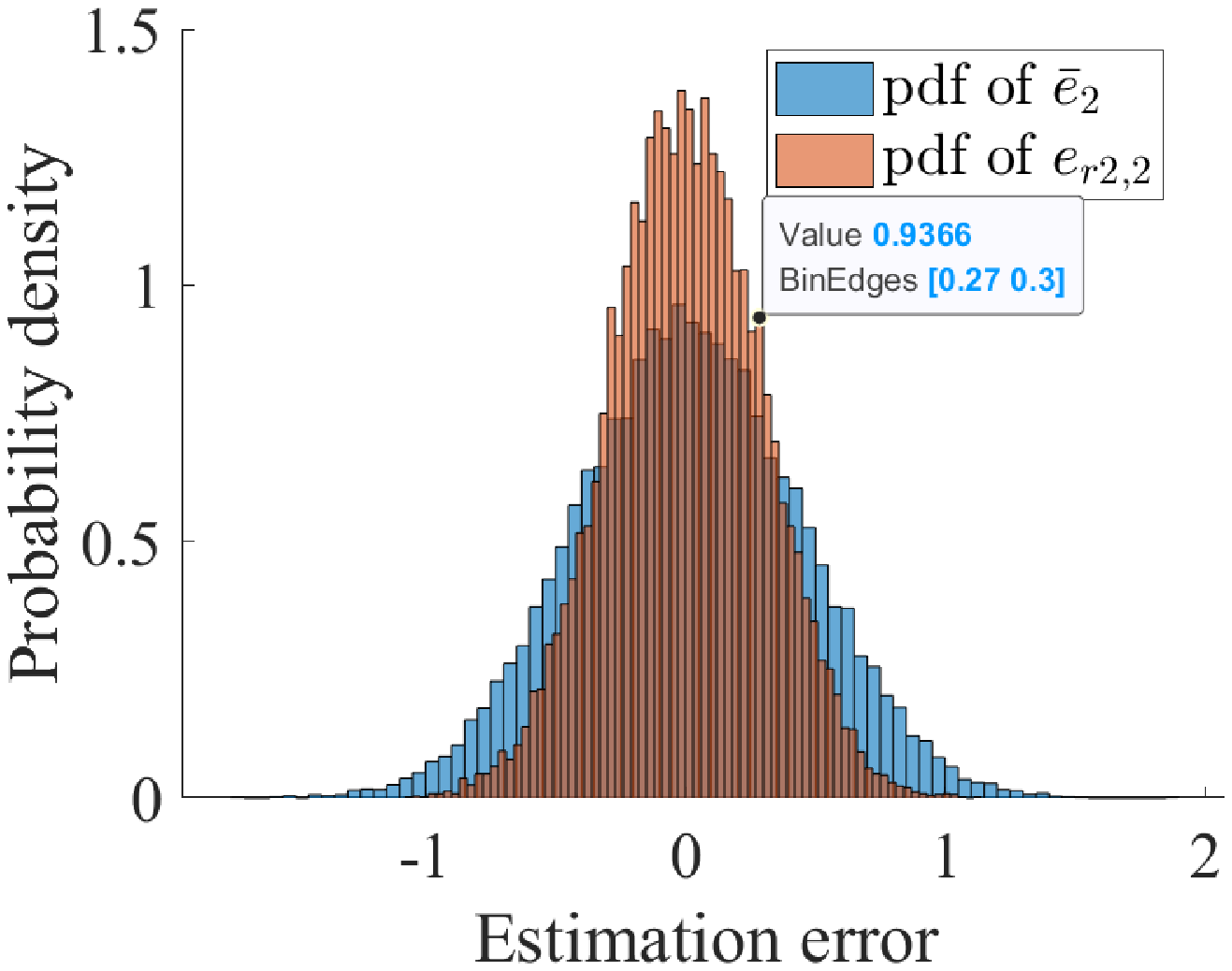}}
	\caption{Probability distributions of the estimate errors of $\bar{C}$ and $C_{r2}$.}
	\label{fig:e2}
\end{figure}

\subsection{The effectiveness of Algorithm 1}

In this subsection, we provide a simulation example to demonstrate the effectiveness of Algorithm 1. We consider the original sensor network $\bar{C}$ defined by \eqref{eq:sim1} and the sensor network with redundant sensors $C_{r2}$ defined by \eqref{eq:NRS12}. Our goal is to use Algorithm 1 to determine the optimal design of redundant sensors, specifically, to find two redundant sensors \textcolor{black}{ $\tilde{C}^{*} = \left[\begin{array}{cc}
	\tilde{C}^{*T}_{1}&\tilde{C}^{*T}_{2}
\end{array}\right]^{T}$} to add to the original sensor network, and then compare the estimation performance of the resulting network \textcolor{black}{ $C^{*}_{r}=\left[\begin{array}{cc}\bar{C}^{T}&\tilde{C}^{*T}
\end{array}\right]^{T}$} with that of $C_{r2}$ to demonstrate the effectiveness of Algorithm 1. Additionally, we illustrate the convergence of Algorithm 1.

Set $C_{r,0} = \left[\begin{array}{cc}
	3 & 0 \\ 3 & 3
\end{array}\right]$, $\epsilon = 0.00001$, $\tilde{R} = I_{2}$, and  $\|\tilde{C}^{j}_{1}\|_{2}\leq5$ , $\|\tilde{C}^{j}_{2}\|_{2}\leq5$. Then, implementing Algorithm 1, the following results can be obtained

\textcolor{black}{
\begin{equation}
\tilde{C}^{*} = \left[\begin{array}{cc}
4.4428&-2.2938\\ 2.0562& 4.5575
\end{array} \right], \gamma^{*}= 0.5572.
\end{equation}}

\begin{figure}[h]
	\centering
	\includegraphics[width=0.7\linewidth]{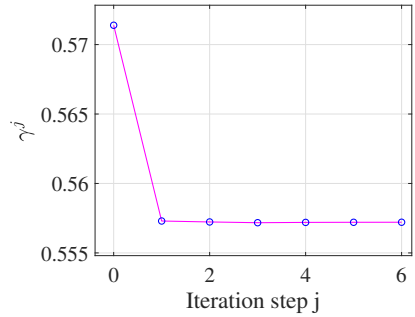}
	\caption{The evolution of $\gamma^{j}$ with iteration step $j$.}
	\label{fig:gammaj}
\end{figure}
\figurename~\ref{fig:gammaj} shows how $\gamma^{j}$ varies with each iteration, which indicates that Algorithm 1 converges very fast. 

\textcolor{black}{
 Implement Kalman filter 20000 steps using $C_{r}^{*}$. \figurename~ \ref{fig:eo2} shows the probability density distribution histograms of $e_{r2}$ and $e_{r*}$. It is observed that the variance of each element in $e_{r*}$ is smaller than that of $e_{r2}$, indicating that the sensor network with redundant sensors' configuration provided by Algorithm 1 achieves better performance.}
 
\textcolor{black}{
To illustrate the scalability of Algorithm 1. Algorithm 1 is executed under varying numbers of redundant sensors using MATLAB(R2019a) with the YALMIP toolbox on an Intel(R) Core(TM) i5-9400 CPU @2.90GHz and 8.00 GB RAM. The average running time over 10 executions under 10 to 1000 redundant sensors are listed in \tablename ~1.}
\begin{table}[htpb]
	\label{tab:Running Time}
	\centering
	\caption{\textcolor{black}{Average Running Times Over ten Executions} }
\textcolor{black}{	\begin{tabular}{ccccc}
		\hline
		Number of $\tilde{C}$ rows	& 10  & 100 & 500 & 1000  \\
		\hline
		 Running Time (s)	&  0.1240  & 1.1900 & 41.3100 & 392.8320\\
		\hline
	\end{tabular}}
\end{table}
\begin{figure}[htpb]
	\centering
	\subfigure[Probability distributions of the  ${e}_{r2,1}$ and  $e_{r*,1}$.]{
		\label{Fig.sub:eo21}
		\includegraphics[width=0.23\textwidth]{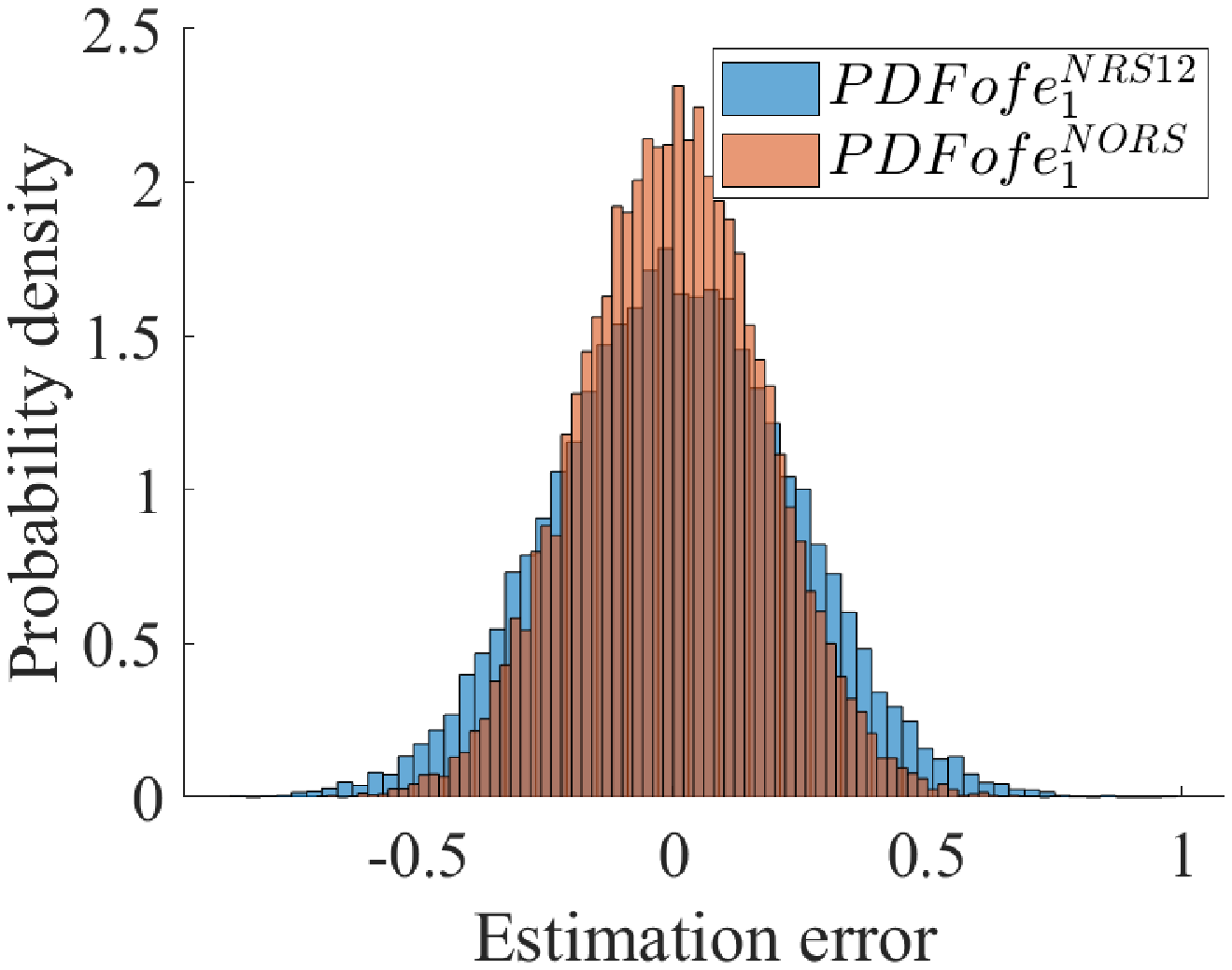}}
	\subfigure[Probability distributions of the  ${e}_{r2,2}$ and  $e_{r*,2}$.]{
		\label{Fig.sub:eo22}
		\includegraphics[width=0.23\textwidth]{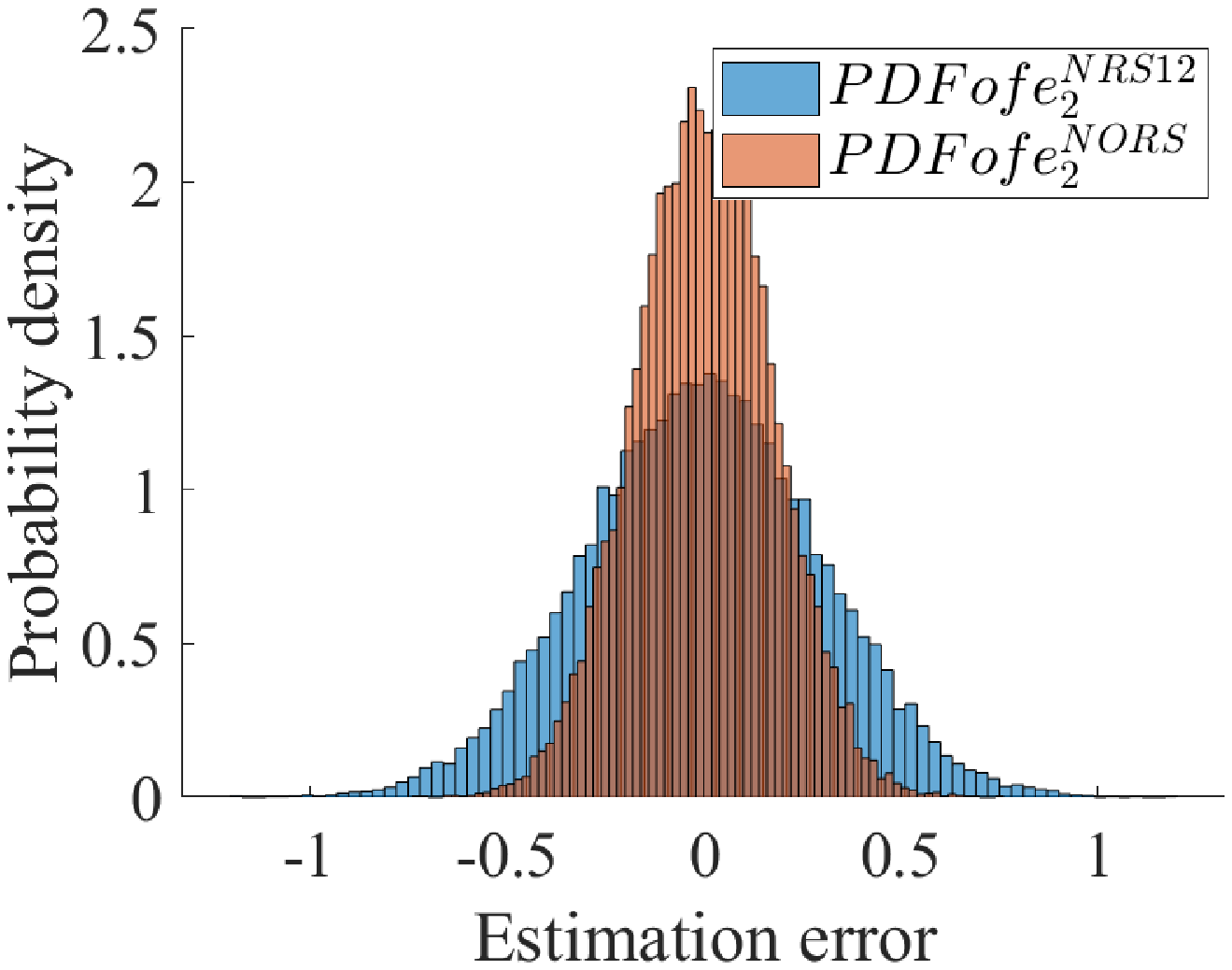}}
	\caption{Probability distributions of the estimate errors of $C_{r2}$ and $C^{*}_{r}$.}
	\label{fig:eo2}
\end{figure}


\section{Conclusion}
This paper delves into a thorough investigation of the effects and optimal design of redundant sensors in \textcolor{black}{collaborative state estimation. }The findings highlight the positive impact of incorporating redundant sensors in enhancing estimation performance. However, it is crucial to note that the benefits of redundant sensors may not extend uniformly to all elements in the state. To address this, the paper introduces a valuable condition that determines the scenarios in which adding redundant sensors can lead to simultaneous improvements in estimating all elements of the system state.  Furthermore, the paper presents an algorithmic approach to optimize the design of redundant sensors, ensuring convergence and optimality. Numerical simulations are conducted to demonstrate the accuracy of the results and the effectiveness of the algorithm. Future research may aim to extend these findings to nonlinear filters as well as investigate other optimization algorithms or strategies that could further enhance the performance of redundant sensors in sensor networks. Such endeavors would contribute to the advancement of sensor network technology and enable improved state estimation and performance across diverse application domains.
%
%



\bibliographystyle{ieeetr}
\bibliography{ref}

\end{document}